\documentclass[12pt, draftclsnofoot, onecolumn]{IEEEtran}

\usepackage{amsthm}
\usepackage[cmex10]{amsmath}
\usepackage{graphicx}
\usepackage{color}
\usepackage[tight,footnotesize]{subfigure}
\usepackage{amsfonts}
\usepackage{algorithmic}
\usepackage{algorithm}
\usepackage{graphicx}
\usepackage{subfigure}
\usepackage{bm}
\usepackage{times,color,amssymb,epsfig,psfrag,stfloats,enumerate}
\usepackage{cite}

\newtheorem{theorem}{Theorem}

\newtheorem{lemma}{Lemma}
\newtheorem{corollary}{Corollary}
\newtheorem{proposition}{Proposition}
\newtheorem{definition}{Definition}

\begin{document}
%
\title{Sleep, Sense or Transmit: Energy-Age Tradeoff for Status Update with Two-Thresholds Optimal Policy}

\author{Jie~Gong,~\IEEEmembership{Member,~IEEE}, Jianhang~Zhu, Xiang~Chen,~\IEEEmembership{Member,~IEEE}, and Xiao~Ma,~\IEEEmembership{Member,~IEEE}

\thanks{J. Gong, J. Zhu and X. Ma are with the School of Computer Science and Engineering, and the Guangdong Key Laboratory of Information
Security Technology, Sun Yat-sen University, Guangzhou 510006, China. Emails: gongj26@mail.sysu.edu.cn, zhujh26@mail2.sysu.edu.cn, maxiao@mail.sysu.edu.cn.}
\thanks{X. Chen is with the School of Electronics and Information Engineering,  Sun Yat-sen University, Guangzhou 510006, China. Email: chenxiang@mail.sysu.edu.cn.}}

\maketitle

\begin{abstract}
Age-of-Information (AoI), or simply age, which measures the data freshness, is essential for real-time Internet-of-Things (IoT) applications. On the other hand, energy saving is urgently required by many energy-constrained IoT devices. This paper studies the energy-age tradeoff for status update from a sensor to a monitor over an error-prone channel. The sensor can sleep, sense and transmit a new update, or retransmit by considering both sensing energy and transmit energy. An infinite-horizon average cost problem is formulated as a Markov decision process (MDP) with the objective of minimizing the weighted sum of average AoI and average energy consumption. By solving the associated discounted cost problem and analyzing the Markov chain under the optimal policy, we prove that there exists a threshold optimal stationary policy with only two thresholds, i.e., one threshold on the AoI at the transmitter (AoIT) and the other on the AoI at the receiver (AoIR). Moreover, the two thresholds can be efficiently found by a line search. Numerical results show the performance of the optimal policies and the tradeoff curves with different parameters. Comparisons with the conventional policies show that considering sensing energy is of significant impact on the policy design, and introducing sleep mode greatly expands the tradeoff range.
\end{abstract}

\begin{IEEEkeywords}
Age-of-information, sleep mode, energy-age tradeoff, Markov decision process
\end{IEEEkeywords}

%
\IEEEpeerreviewmaketitle

\section{Introduction}
With the continuous increase of real-time Internet-of-Things (IoT) applications such as remote monitoring and control, phase data update in smart grid, environment monitoring for autonomous driving and etc., timely status information is strictly required to guarantee a fast and accurate response~\cite{Abd19on}. Thus, it is essential to persistently obtain fresh data. The \emph{age-of-information} (AoI), or simply \emph{age}, defined as the time elapsed since the generation of the latest received update~\cite{kaul2012real}, is a candidate performance metric for data freshness. Different from the conventional delay performance, AoI captures the impact of both transmission delay and data generation frequency.

On the other hand, the sensors gathering status information are usually energy-constrained. Thus, it is also vital to reduce the sensors' energy consumption. However, the two objectives, keeping data fresh and reducing energy consumption, can not be achieved simultaneously in general. To minimize AoI, the sensors should sense and transmit new status in time. To reduce energy consumption, ``lazy" policy is preferred, i.e., the sensors may sense with a low frequency and turn to sleep mode. Therefore, there exists a fundamental tradeoff between AoI and energy consumption, which is of great significance to provide a guidance to determine the sensing and transmission policy in energy-constrained status update communications.

In this paper, to exploit the energy-age tradeoff, we consider a status update system where a sensor generates and transmits status packets to a monitor over an error-prone channel. In this system, a data transmission may fail due to a channel error, and the sensor is aware of the transmission success/failure via ACK/NACK feedback protocol. If the transmission succeeds, the data at the monitor is refreshed and the AoI is reduced. As both status sensing and data transmission consume energy, a fundamental problem is when to sleep, sense, and transmit to balance data freshness and energy consumption. Intuitively, when a fresh data is successfully transmitted, the sensor may sleep to save energy. If the data becomes stale on the contrary, new data should be sensed and transmitted to reduce AoI. The main goal of this paper is to find the optimal policy to achieve the tradeoff between energy and age.

\subsection{Related Work}
The energy and age related research is originated from \cite{yates2015lazy}, where an energy harvesting source was considered. In this early work, \emph{zero-wait policy} was introduced, which generates a fresh update just as the prior update is delivered and the channel becomes idle. The optimality of zero-wait policy was analyzed in \cite{sun2017update}. Then, the age-optimal policies were extensively studied by assuming infinite battery \cite{bac2015age}, unit battery \cite{bac2017scheduling} and finite battery~\cite{bac2019optimal, arafa2020age}, respectively. Transmit power control under an energy harvesting constraint was introduced in \cite{arafa2017age}. With random arrival of status updates and energy units, the average AoI performance with finite battery was analyzed in \cite{farazi2018age, farazi2018average}. {Reinforcement learning approach was adopted to minimize the average AoI for a single energy harvesting sensor~\cite{Ceran19rein} and multiple wireless powered sources~\cite{abd20a}, respectively. Joint sampling and updating for wireless powered communication systems considering time and energy costs for sensing was studied in \cite{abd20aoi}. Peak AoI minimization with random energy arrivals was studied in~\cite{ozel20timely} considering both sensing and transmission energy.} The above works usually aim to optimize the AoI under a certain energy constraint. Different from them, in energy-age tradeoff problem, age can be sacrificed to reduce energy and vice versa.

Status update over error-prone channels has been extensively studied in the literature. Timely updates over an erasure channel for infinite incremental redundancy and fixed redundancy were considered in \cite{yates2017timely}. The optimal status update policy without feedback was studied in \cite{arafa2018online}. With ACK/NACK feedback, automatic repeat-request (ARQ) and hybrid ARQ (HARQ) protocols were adopted to keep the data fresh \cite{chen2016age, najm2017status, parag2017on, sac2018age, arafa2019on, ceran2019average, li2020age, huang2020real}. The most relevant paper to our work is \cite{ceran2019average}, where the age-optimal policy under a resource constraint was obtained for ARQ and HARQ protocols. However, the sensing energy is not considered, which is non-negligible in many applications and is even larger than the transmit energy (see \cite{liu2016energy} and references therein). More importantly, the consideration of sensing energy in policy design will induce significant difference. Specifically, if the sensing energy is not considered, sensing and transmitting a new packet can always reduce AoI without additional cost compared with retransmission. On the contrary, if the sensing energy is non-negligible, there exists a fundamental tradeoff between AoI reduction by sensing a new data and energy saving by retransmission. Thus, different from \cite{ceran2019average}, one of the main challenges to be tackled in this paper is how to balance the energy cost between sensing and transmission.

Recently, the energy-age tradeoff analysis has drawn more and more attention. The energy-age tradeoff in an error-prone channel was revealed in \cite{gong2018energy}. Then, the analysis was extended to the fading channel \cite{huang2020age}. The impact of coding on the tradeoff performance was analyzed in~\cite{Xie2020age}. Without ACK/NACK feedback, the energy-age tradeoff for random update generation was also found in \cite{gu2019timely}. The tradeoff between energy efficiency and AoI in a multicast system was optimized in \cite{nath2018optimum}. However, sleep mode is not considered in these works, although it is promising for energy saving as was shown in the conventional studies on energy-delay tradeoff~\cite{wu2013traffic, niu2015char, wu2016base}. Since age and delay are closely related, it is crucial to consider sleep mode, which however, makes the study on the energy-age tradeoff very challenging.

\subsection{Main Results}
Based on Markov decision process (MDP) \cite{bertsekas2005dynamic}, an infinite-horizon average cost problem is formulated with the objective of minimizing a weighted sum of average AoI and average energy consumption. The sensor stores only the latest sensed data packet. It can either sleep, retransmit the stored packet, or sense and transmit a new one depending on the ACK/NACK feedback and the instantaneous AoI. To characterize the system state, the AoI at the transmitter side (AoIT) and the AoI at the receiver side (AoIR) are introduced. {We show that there exists a threshold optimal policy to solve the MDP problem. It is remarkable that threshold optimal structure has been found in many works (e.g., \cite{Ceran19rein, abd20a, abd20aoi, hsu20scheduling}). However, how to find the optimal thresholds is quite challenging as they are usually state-dependent. Different from the existing works, we show that the optimal policy is determined by only two thresholds, and one of them can be expressed by the other in closed-form. Moreover, these two thresholds can be efficiently found by a line search.} In summary, the main contributions are listed as follows.
\begin{itemize}
\item {We prove that there exists a \emph{two-thresholds} optimal stationary policy for the average cost problem. In particular, given the optimal thresholds $\theta_{\mathrm t}$ and $\theta_{\mathrm r}$, the sensor
\begin{itemize}
\item senses and transmits a new packet if AoIT~$\ge \theta_{\mathrm t}$ and AoIR~$\ge \theta_{\mathrm r}$,
\item re-transmits the stored packet if AoIT~$<\theta_{\mathrm t}$ and AoIR~$\ge \theta_{\mathrm r}$,
\item stays in sleep mode if AoIR~$<\theta_{\mathrm r}$.
\end{itemize}}
\item To obtain the above result, we firstly show that the optimal policy for the average cost problem is a limit of the optimal policies for the associated discounted cost problem with discount factors tending to 1. Then, we prove that there exists a threshold optimal stationary policy for the discounted problem. Thus, the optimal policy for the average cost problem is also threshold-based. Finally, we show that the minimum average cost can be expressed in closed-form by only two thresholds, which results in the two-thresholds optimal policy.
\item Furthermore, we show that $\theta_{\mathrm r}$ can be expressed as a closed-form function of $\theta_{\mathrm t}$, and the optimal $\theta_{\mathrm t}$ is upper bounded. Accordingly, an efficient line search algorithm is proposed to find the optimal thresholds.
\item We characterize the energy-age tradeoff under the optimal policy via numerical simulations, and show the impact of system parameters on the tradeoff curves. We also compare with the existing policies to show the performance gain of our policy.
\end{itemize}

The rest of this paper is organized as follows. The system model and problem formulation is given in Sec.~\ref{sec:model}. The associated discounted cost problem is presented in Sec.~\ref{sec:relate}. Sec.~\ref{sec:threshold} proves the optimal threshold policy, and the optimal thresholds are obtained in Sec.~\ref{sec:solve}. Numerical results are shown in Sec.~\ref{sec:simulation}. Finally, Sec.~\ref{sec:conclusion} concludes the paper.

\section{System Model and Problem Formulation} \label{sec:model}
Consider a status update system in Fig.~\ref{fig:system}, where a sensor generates and transmits status packets to a monitor over an error-prone channel. Assume the system is slotted, and without loss of generality, the slot length is set to 1. At the beginning of each slot, the sensor can either sense or not. If it senses, a new packet conveying the latest status is generated. Denote the sensing energy consumption by $E_{\mathrm s}$, and assume the time for sensing is negligible. {The assumption is reasonable in many sensing applications such as taking real-time pictures by a camera with high shutter speed.} The sensor only stores the latest sensed data packet. If the sensor does not sense, the last sensed data packet remains there. Otherwise, it is replaced by the newly sensed packet. During each slot, at most one packet can be transmitted. Denote the energy consumption for transmitting a packet by $E_{\mathrm t}$. The transmitted packet may not be successfully received by the monitor due to channel error. Denote the channel error probability by $p \in (0, 1)$, which is identical and independent among slots. If the packet is successfully received by the monitor, an ACK is sent back to the sensor via an error-free feedback channel. Otherwise, a NACK is sent back. Depending on the feedback, the sensor decides its action in the next slot.

\begin{figure}
\centering
\includegraphics[width=3.2in]{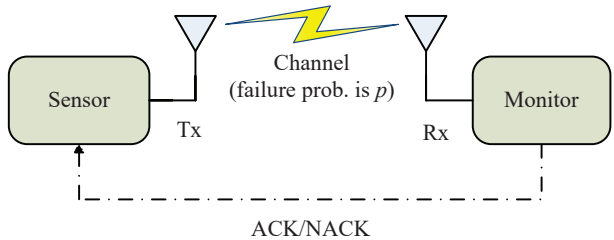}
\caption{Status update system over an error-prone channel with feedback.} \label{fig:system}
\end{figure}

Denote by $s_k = 1$ if a new packet is generated at the beginning of slot $k$, and $s_k = 0$ otherwise. Denote by $t_k = 1$ if the stored packet is transmitted in slot $k$, and $t_k = 0$ otherwise. Thus, the sensor has four possible actions in each slot: (1) sleep (neither sense nor transmit), i.e., $s_k = 0, t_k = 0$; (2) sense the latest status but do not transmit, i.e., $s_k = 1, t_k = 0$; (3) do not sense but retransmit the stored packet, i.e., $s_k = 0, t_k = 1$; (4) sense and transmit a new packet, i.e., $s_k = 1, t_k = 1$. The actions are taken by jointly considering the information freshness and the sensor's energy consumption. {Intuitively, the second action is not a good option as it consumes energy but does not contribute to AoI. In fact, as proved later in Lemma \ref{lemma:action}, the optimal policy does not contain this action.} Next, we give some definitions before the problem formulation.

\subsection{Definitions of AoI}
The definitions of AoI, AoIT, and AoIR are given as follows.
\begin{definition}
At time $t$, if the observed latest packet is generated at time $U(t)$, the AoI is
\begin{align}
\Delta(t) = t - U(t).
\end{align}
\end{definition}
\begin{definition}
The AoIT, denoted by $\Delta_T(t)$, is the AoI observed at the transmitter side.
\end{definition}
\begin{definition}
The AoIR, denoted by $\Delta_R(t)$, is the AoI observed at the receiver side.
\end{definition}

An example of AoIT and AoIR curves is depicted in Fig.~\ref{fig:aoi}. As both AoIT and AoIR grow linearly during each slot, it is sufficient to characterize the AoIT and AoIR curves by the instantaneous AoIs at the beginning of each slot. In particular, denote $a_{\mathrm t, k}$ and $a_{\mathrm r, k}$ as the instantaneous AoIT and AoIR at the beginning of slot $k$ (before taking any action), respectively. {The evolution from $(a_{\mathrm t, k},a_{\mathrm r, k})$ to $(a_{\mathrm t, k+1}, a_{\mathrm r, k+1})$ depends on the action. If the sensor does not generate a new packet, i.e., $s_k = 0$, the packet at the transmitter side is not updated. Hence, the AoIT grows linearly. If a new packet is generated at the beginning of slot $k$, i.e., $s_k = 1$, the AoIT at the beginning of slot $k+1$ becomes 1. Similarly, the packet at the receiver side is not updated if the sensor does not transmit ($t_k = 0$) or transmits ($t_k = 1$) but fails. In this case, the AoIR grows linearly. If the transmission succeeds, the AoIR becomes the same as the AoIT because the data at both sides are the same. In summary, we have}
\begin{align}
(a_{\mathrm t, k+1}, a_{\mathrm r, k+1}) = \left\{\begin{array}{ll}
(a_{\mathrm t, k}\!+\!1, a_{\mathrm r, k}\!+\!1), &\textrm{if~} s_k\! =\! t_k \!=\! 0, \textrm{~or~} s_k \!= \!0, t_k \!=\! 1 \textrm{~and~fails},\\
(a_{\mathrm t, k}\!+\!1, a_{\mathrm t, k}\!+\!1), &\textrm{if~} s_k \!=\! 0, t_k \!=\! 1 \textrm{~and~succeeds},\\
(1, a_{\mathrm r, k}\!+\!1), &\textrm{if~} s_k \!=\! 1, t_k \!=\! 0, \textrm{~or~} s_k \!=\! t_k \!=\! 1 \textrm{~and~fails},\\
(1, 1), &\textrm{if~} s_k \!=\! 1, t_k \!=\! 1 \textrm{~and~succeeds}.
 \end{array} \right.
\end{align}


\begin{figure}
\centering
\includegraphics[width=4in]{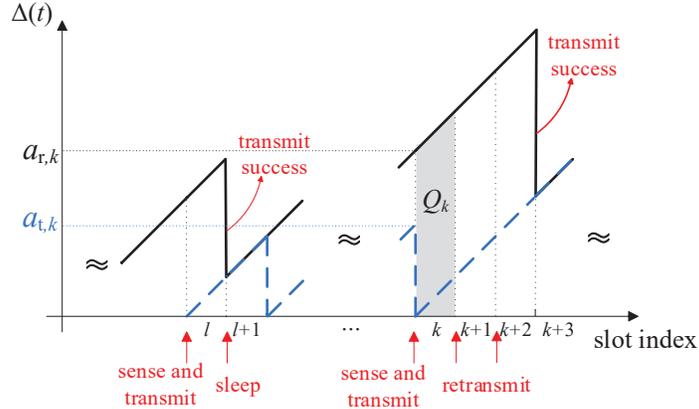}
\caption{Evolution of AoIT $\Delta_T(t)$ (blue dashed curve) and AoIR $\Delta_R(t)$ (black bold curve).} \label{fig:aoi}
\end{figure}

\subsection{Energy-Age Tradeoff Problem}
To maintain a low AoIR, a new packet should be generated and transmitted in every slot. To reduce the energy consumption, on the other hand, the sensor should reduce the number of sensing and transmission actions. Therefore, there is a tradeoff between energy and age.

From the age perspective, the average AoIR is defined as \cite{kaul2012real}
\begin{align}
\bar{\Delta} = \lim_{\tau \rightarrow \infty} \frac{1}{\tau} \int_{0}^{\tau} \Delta_R(t) \mathrm{d} t.
\end{align}
It can be geometrically computed as the sum of the area of the trapezoid $Q_k$ as in Fig.~\ref{fig:aoi} divided by the total time. As the slot length equals to one, we have
\begin{align}
\bar{\Delta} = \lim_{K\rightarrow\infty} \frac{1}{K} \mathbb{E}\left[\sum_{k=1}^{K} Q_k\right] = \lim_{K\rightarrow\infty} \frac{1}{K} \mathbb{E}\left[\sum_{k=1}^{K} \frac{a_{\mathrm r, k} + (a_{\mathrm r, k} \!+\! 1)}{2}\right] = \lim_{K\rightarrow\infty} \frac{1}{K} \mathbb{E}\left[\sum_{k=1}^{K} a_{\mathrm r, k}\right] + \frac{1}{2},
\end{align}
where the expectation $\mathbb{E}$ is taken over all possible channel realizations over all the slots.

The average energy consumption of the sensor can be calculated as
\begin{align}
\bar{E} = \lim_{K\rightarrow\infty} \frac{1}{K} \mathbb{E}\left[\sum_{k=1}^{K} (s_k E_{\mathrm s} + t_k E_{\mathrm t})\right].
\end{align}

Our objective is to exploit the tradeoff between the average AoIR and the average energy consumption. {Similar to \cite{berry02comm}}, we consider minimizing the weighted sum of the average AoIR and the average energy consumption by deciding when to sleep, sense and transmit, i.e.,
\begin{align}
\min_{\{s_k, t_k\}_{k=1}^{\infty}} \lim_{K\rightarrow\infty} \frac{1}{K} \mathbb{E}\left[\sum_{k=1}^{K} \left(a_{\mathrm r, k} + \omega (s_k E_{\mathrm s} + t_k E_{\mathrm t})\right)\right] + \frac{1}{2}, \label{prob:weightMin}
\end{align}
where $\omega > 0$ is a pre-defined weighting factor indicating the importance of energy over age. By solving \eqref{prob:weightMin} for a set of $\omega$'s, the tradeoff pairs of age and energy can be obtained. {In practice, the weighting factor can be dynamically tuning depending on the available energy.} In the rest of this paper, the constant factor $\frac{1}{2}$ in \eqref{prob:weightMin} is ignored as it does not influence the policy design.

{
\textbf{Remark 1.}
We should emphasize here that the proposed solution for the above optimization problem \eqref{prob:weightMin} can be directly applied to the multi-sensor IoT systems when the sensors are allocated with orthogonal frequency bands. It can also be extended to the case that the sensors are scheduled by the round-robin algorithm. In particular, our result can be applied to each sensor by redefining a slot as a scheduling round. As the actual AoI is a linear function of the redefined slots since generation, the structure of the proposed solution does not change.

\textbf{Remark 2.}
The energy-age tradeoff can also be characterized by considering the constrained problem $\min\limits_{\bar{E} \le E_{\max}} \bar{\Delta}$. In fact, it can be solved based our solution. In particular, if there exists an $\omega^*$ such that the corresponding average energy consumption in the minimum weighted sum is $E_{\max}$, the optimal policy to minimize the weighted sum also minimizes $\bar{\Delta}$ under constraint $\bar{E} \le E_{\max}$. Otherwise, there exists a randomized policy to solve the constrained problem that is a combination of the optimal policies for two weighted sum problems, one with the corresponding energy consumption smaller than $E_{\max}$, and the other with the energy larger than $E_{\max}$ \cite{sennott93constr}.
}

\subsection{MDP Problem Formulation}
The problem \eqref{prob:weightMin} can be reformulated as an MDP. A general MDP is characterized by four key components: state, action, state transition and per-stage cost. Firstly, the stage in our problem refers to the time slot. For ease of description, slot and stage can be used interchangeably in the rest of this paper. The state in stage $k$ includes both AoIT and AoIR at the beginning of slot $k$, which can be denoted by $x_k = (a_{\mathrm t, k}, a_{\mathrm r, k}) \in \mathcal{A}$, where the state space can be denoted by $\mathcal{A} = \{(i,j)|i \le j, i,j \in \{1,2,\cdots\}\}$. Notice that we have $a_{\mathrm t, k} \le a_{\mathrm r, k}$ as the packet at the receiver is always ``older" than the one at the transmitter.

The action to be taken includes sensing decision and transmission decision, i.e., $u_k = (s_k, t_k) \in \mathcal{U}$, where the action space can be denoted by $\mathcal{U} = \{(0,0), (0,1), (1,0), (1,1)\}$. 

The state transition probability $p_{(i,j)\rightarrow (i',j')}(u) = \mathrm{Pr}(x_{k+1} = (i',j')|x_k = (i,j), u_k = u)$ can be calculated as
\begin{align}
p_{(i,j)\rightarrow (i',j')}(u) =\left\{ \begin{array}{ll} 1, & \textrm{if~} i \le j, i'=i+1, j'=j+1, u = (0,0), \\
{}&\quad\textrm{or~} i = j, i'=j'=i+1, u = (0,1), \\
{}&\quad\textrm{or~} i \le j, i'=1, j'=j+1, u = (1,0),\\
p, &\textrm{if~} i < j, i'=i+1, j'=j+1, u = (0,1), \\
{}&\quad\textrm{or~} i \le j, i'=1, j'=j+1, u = (1,1),\\
1-p, & \textrm{if~} i < j, i'=j'=i+1, u = (0,1), \\
{}&\quad\textrm{or~} i \le j, i'=j'=1, u = (1,1),\\
0, &\textrm{else.}\end{array} \right. \label{eq:statetrans}
\end{align}

The cost per stage can be expressed as
\begin{align}
g(x_k, u_k) = a_{\mathrm r, k} + \omega(s_k E_{\mathrm s} + t_k E_{\mathrm t}). \label{eq:stagecost}
\end{align}

Thus, the problem \eqref{prob:weightMin} can be reformulated as
\begin{align}
\min_{\pi} \; J_{\pi}(x_0) \overset{\textrm{def}}{=} \lim_{K\rightarrow \infty} \frac{1}{K} \mathbb{E}\left( \sum_{k=0}^{K-1} g(x_k, \mu_k(x_k)) \Bigg| x_0\right), \label{prob:avgmin}
\end{align}
where $x_0$ is the initial state, and $\pi = \{\mu_0, \mu_1, \cdots\}$ is the policy. In general, $\mu_k$ is possibly randomized and follows a conditional probability distribution of the action $u_k$ under the condition of a given state $x_k$, i.e. $\mu_k \sim p_k(u_k|x_k)$, where $0\le p_k(u_k|x_k) \le 1$ and $\sum_{u_k\in\mathcal{U}} p_k(u_k|x_k) = 1, \forall x_k$. A policy is called \emph{stationary} if the distribution $p_k(u_k|x_k)$ does not change over stage $k$, which can be denoted as $\mu_k =\mu \sim p(u_k|x_k), \forall k$. Furthermore, a policy is \emph{deterministic} if $p_k(u_k|x_k) \in \{0, 1\}$. In this case, $\mu_k:\mathcal{A} \rightarrow \mathcal{U}$ is a mapping from the state space to the action space with $u_k = \mu_k(x_k)$ if $p_k(u_k|x_k) = 1$.

It is remarkable that the state space $\mathcal{A}$ is countably infinite and the per-stage cost is unbounded. Hence, the conventional iterative algorithms such as policy iteration or value iteration are difficult to be applied directly in practice. To tackle this difficulty, we next exploit the structural properties of our problem to find the optimal policy. To this end, we firstly associate the average cost minimization problem with its discounted version. Then, we prove the structural properties for the discounted problem, which turn out to hold for the average cost problem as well.

\section{Associated Discounted Cost Problem} \label{sec:relate}
It is well-known that the average cost MDP problem is closely related to its discounted version~\cite{sennott1989average}. Thus, we start by considering an infinite horizon discounted cost MDP problem as
\begin{align}
\min_{\pi} \; J_{\alpha, \pi}(x_0) \overset{\textrm{def}}{=} \lim_{K\rightarrow \infty} \mathbb{E}\left( \sum_{k=0}^{K-1} \alpha^kg(x_k, \mu_k(x_k)) \Bigg| x_0\right), \label{prob:discount}
\end{align}
where $\alpha \in (0,1)$. Since the AoIR grows at most linearly versus the stage $k$ and $g(x_k, \mu_k) \le a_{\mathrm r, k} + \omega(E_{\mathrm s} + E_{\mathrm t})$, we have
\begin{align}
J_{\alpha, \pi}(x_0) \le \sum_{k=0}^{\infty} \alpha^k (a_{\mathrm r, 0} + k + \omega(E_{\mathrm s} + E_{\mathrm t})) = \frac{1}{1-\alpha}\left( a_{\mathrm r, 0} + \omega(E_{\mathrm s} + E_{\mathrm t}) + \frac{\alpha}{1-\alpha}\right) < \infty. \label{eq:finite}
\end{align}
Hence, for any given initial state $x_0$ and discount factor $\alpha$, $J_{\alpha, \pi}(x_0)$ is well-defined for all policies. Given the initial state $x_0 = (i,j)$, denote the minimum expected discounted cost as
\begin{align}
J_{\alpha}(i,j) = J_{\alpha}(x_0) = \min_{\pi} J_{\alpha, \pi}(x_0). \label{eq:discount}
\end{align}


In the following, we will show the existence of an optimal stationary deterministic policy for the discounted cost problem \eqref{prob:discount} and its convergence to the average cost problem \eqref{prob:avgmin}. Thus, if there is a threshold optimal policy for the discounted version, the optimal policy for the average cost problem is also threshold-based.

For an infinite horizon discounted cost problem with finite state space and bounded per-stage cost, it is guaranteed that the optimal discounted cost can be achieved by dynamic programming (DP) algorithm \cite[Prop.~1.2.1]{bertsekas2005dynamic}, and there exists a stationary deterministic policy to attain the minimum \cite[Prop.~1.2.3]{bertsekas2005dynamic}. However, as the problem \eqref{prob:discount} has a countably infinite state space and unbounded per-stage cost, the convergence of the DP algorithm and the existence of a stationary policy need to be reconsidered. In fact, observing that the AoIR increases at most linearly with stage $k$, the convergence of the DP algorithm can still be proved by modifying some lines in the proof of \cite[Prop.~1.2.1]{bertsekas2005dynamic}, and there is also a stationary deterministic policy. The results are summarized as follows.

\begin{proposition} \label{prop:discBellman}
Define a sequence of functions for all $x \in \mathcal{A}$ as
\begin{align}
J_{n}(x) &= \min_{u \in \mathcal{U}}\left( g(x, u) + \alpha \sum_{x'\in\mathcal{A}}p_{x\rightarrow x'}(u)J_{n-1}(x') \right), \quad n \ge 1 \label{eq:Jalphan}
\end{align}
where $J_0(x) = 0$, $g(x,u)$ is given as \eqref{eq:stagecost} and $p_{x\rightarrow x'}(u)$ is given as \eqref{eq:statetrans}. We have
\begin{align}
\lim_{n\rightarrow\infty} J_{n}(x) = J_{\alpha}(x), \label{eq:converge}
\end{align}
where $J_{\alpha}(x)$ is given as \eqref{eq:discount}. In addition, the Bellman's Equation holds as
\begin{align}
J_{\alpha}(x) = \min_{u \in \mathcal{U}} \left( g(x, u) + \alpha \sum_{x'\in\mathcal{A}}p_{x\rightarrow x'}(u)J_{\alpha}(x') \right), \label{eq:bellman}
\end{align}
and the stationary policy $\mu_{\alpha}$, where $u = \mu_{\alpha}(x)$ attaining the minimum of \eqref{eq:bellman}, is optimal.
\end{proposition}
\begin{proof}
See Appendix \ref{proof:discBellman}.
\end{proof}


The following lemma shows that $J_{n}(x)$ is monotonic versus each component of $x$.
\begin{lemma} \label{lemma:monoton}
For any $n \ge 1$ and $x = (i,j)\in \mathcal{A}$, the function $J_{n}(x)$ defined in \eqref{eq:Jalphan} satisfies
\begin{align}
J_{n}(i, j+1) - J_{n}(i, j) &\ge 1 , \label{eq:nmonoj}\\
J_{n}(i+1, j) - J_{n}(i, j) &\ge 0. \label{eq:nmonoi}
\end{align}
\end{lemma}
\begin{proof}
See Appendix \ref{proof:monoton}.
\end{proof}

Now, we relate the average cost minimization problem to its discounted version. It is shown in \cite{sennott1989average} that a stationary policy for an average cost problem exists under some conditions, and can be obtained by letting $\alpha$ tend to 1 in the associated discounted cost problem. According to the main results and conditions given in \cite{sennott1989average}, we can have the following proposition.

\begin{proposition} \label{prop:converge}
Let $\alpha_n$ be any sequence of discount factors converging to 1 with the associated optimal stationary policy $\mu_{\alpha_n}$ for the discounted cost problem \eqref{prob:discount}. There exists a subsequence of $\alpha_n$ denoted by $\beta_n$ and a stationary policy $\mu$ that is a limit of $\mu_{\beta_n}$. That is, for every state $x$, there exists an integer $N(x)$ such that $\mu_{\beta_n}(x) = \mu(x)$ for all $n \ge N(x)$. In addition, the stationary policy $\mu$ is optimal for the average cost minimization problem \eqref{prob:avgmin}.
\end{proposition}
\begin{proof}
See Appendix \ref{proof:converge}.
\end{proof}
Based on Proposition \ref{prop:converge}, if some structural property of the discounted cost problem can be found, it must hold for the corresponding average cost problem. Furthermore, a quick observation in the following lemma helps to reduce the size of the action space.

\begin{lemma} \label{lemma:action}
The optimal stationary policy $\mu_{\alpha}$ for the problem \eqref{prob:discount} satisfies
\begin{align}
\mu_{\alpha}(x) \neq (1,0), \quad \forall x \in \mathcal{A}
\end{align}
\end{lemma}
\begin{proof}
See Appendix \ref{proof:action}.
\end{proof}

Based on Proposition \ref{prop:converge}, we can further prove that $\mu(x) \neq (1,0), \forall x \in \mathcal{A}$. Therefore, only the actions $(0,0), (0,1)$, and $(1,1)$ need to be considered. For simplicity, in the rest of this paper, the action is re-denoted as $u_k = s_k + t_k$ and the action space is $\mathcal{U} = \{0, 1, 2\}$. Then, the DP algorithm for the problem \eqref{prob:discount} can be explicitly described as follows.
\begin{align}
J_{n}(i,j) &= \min_{u\in \{0,1,2\}}\left\{ Q_{n}(i,j,u)\right\}, \quad n \ge 1, \label{eq:Jn}
\end{align}
where $J_0(i,j) = 0$, and
\begin{align}
Q_{n}(i,j,0) &= j + \alpha J_{n-1}(i+1,j+1), \label{eq:Qn0}\\
Q_{n}(i,j,1) &= j + \omega E_{\mathrm t} + \alpha p J_{n-1}(i+1,j+1) + \alpha (1-p) J_{n-1}(i+1,i+1), \label{eq:Qn1}\\
Q_{n}(i,j,2) &= j + \omega E_{\mathrm t} + \omega E_{\mathrm s} + \alpha p J_{n-1}(1,j+1) + \alpha (1-p) J_{n-1}(1,1). \label{eq:Qn2}
\end{align}
Denote by
\begin{align}
\mu_n(i,j) = \arg \min_{u\in \{0,1,2\}}\left\{ Q_{n}(i,j,u)\right\}
\end{align}
According to Proposition \ref{prop:discBellman}, we have $J_{\alpha}(i,j) = \lim_{n \rightarrow \infty} J_n(i,j), \mu_{\alpha}(i,j) = \lim_{n \rightarrow \infty} \mu_n(i,j)$. Denote by $Q_{\alpha}(i,j,u) = \lim_{n \rightarrow \infty} Q_n(i,j,u), u \in \{0, 1, 2\}$. In the next section, we will show the threshold optimal policy by exploiting the properties of the above DP iteration.

\section{Threshold Optimal Policy} \label{sec:threshold}
To show the structural property of the discounted cost problem, we denote $\mathcal{A}_{\mathrm{ext}} = \{(i,j)|i,j \in \{1, 2, \cdots\}\}$. With the extended state space $\mathcal{A}_{\mathrm{ext}}$ without constraint $i \le j$, the monotonicity in Lemma \ref{lemma:monoton} still holds and the DP iteration \eqref{eq:Jn} still converges to the optimal cost of the problem \eqref{prob:discount} for $(i,j) \in \mathcal{A}$ as the cost function for state $(i,j) \in \mathcal{A}$ does not rely on any state in $\mathcal{A}_{\mathrm{ext}} \setminus \mathcal{A}$. Then, we exploit the properties of $J_n(i,j)$ and $J_{\alpha}(i,j)$ for $(i,j) \in \mathcal{A}_{\mathrm{ext}}$.

\begin{lemma} \label{lemma:i0}
	For any $n \ge 1$, if $i\ge i_{0} \overset{\textrm{def}}{=} \left\lceil \dfrac{\omega E_{\mathrm s}}{\alpha \left(1-p \right) }\right\rceil$ or $i \ge j$, we have
	\begin{align}
		\mu_{n}(i,j)&=\mu_{n}(i+1,j)\in\{0,2\},\label{L3a0}\\
		J_{n}(i,j)&=J_{n}(i+1,j).\label{L3b0}
	\end{align}
\end{lemma}
\begin{proof}
See Appendix \ref{proof:i0}.
\end{proof}

\begin{lemma} \label{lemma:j0}
	For any $n\ge 2$ and $j\ge j_{0}\overset{\text{def}}{=}\left\lceil \dfrac{\omega E_{\mathrm t}+\omega E_{\mathrm s}}{\alpha \left(1-p \right) }\right\rceil$, there exists $i_{n}$ such that
	\begin{align}
		&\mu_{n}(i,j)=\left\{\begin{array}{ll} 1,& i<i_{n}\\2,& i\ge i_{n}\end{array} \right. \label{L4a0}\\
		&J_{n}(i,j+1)-J_{n}(i,j)=\dfrac{1-\left(\alpha p\right)^n }{1-\alpha p},\label{L4c0}
	\end{align}
\end{lemma}
\begin{proof}
See Appendix \ref{proof:j0}.
\end{proof}
Based on Lemma \ref{lemma:i0} and Lemma \ref{lemma:j0}, as $J_{\alpha}(i,j) = \lim_{n \rightarrow \infty} J_n(i,j), \mu_{\alpha}(i,j) = \lim_{n \rightarrow \infty} \mu_n(i,j)$, similar properties hold for the optimal policy and the optimal cost. The results are summarized as follows.
\begin{corollary}\label{coro:i0j0}
If $i\ge i_{0}$ or $i \ge j$, we have
	\begin{align}
		\mu_{\alpha}(i,j)&=\mu_{\alpha}(i+1,j)\in\{0,2\},\label{L3a}\\
		J_{\alpha}(i,j)&=J_{\alpha}(i+1,j).\label{L3b}
	\end{align}
If $j\ge j_{0}$, there exists $i_{\alpha}$ such that
	\begin{align}
		&\mu_{\alpha}(i,j)=\left\{\begin{array}{ll} 1,& i<i_{\alpha}\\2,& i\ge i_{\alpha}\end{array} \right. \label{L4a}\\
		&J_{\alpha}(i,j+1)-J_{\alpha}(i,j)=\dfrac{1}{1-\alpha p}.\label{L4c}
	\end{align}
\end{corollary}

\begin{figure}
  \centering
  \includegraphics[width=3in]{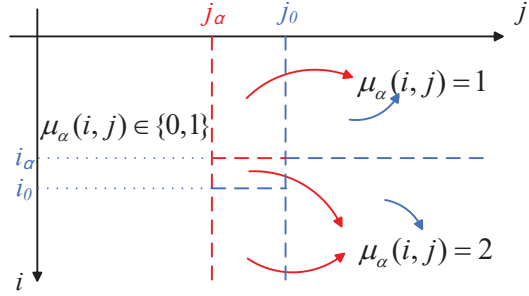}\\
  \caption{The structure of the optimal policy for the discounted cost problem.}\label{fig:optproof}
\end{figure}

Based on Corollary \ref{coro:i0j0}, as $\mu_{\alpha}(i,j)\in \{0,2\}$ for $i\ge i_{0}$ and $\mu_{\alpha}(i,j)\in \{1,2\}$ for $j\ge j_{0}$, we have $\mu_{\alpha}(i,j)=2$ for $i\ge i_{0},j\ge j_{0}$. Hence, according to \eqref{L4a}, we have $i_{\alpha} \le i_{0}$, and $\mu_{\alpha}(i,j)=2$ for $i\ge i_{\alpha},j\ge j_{0}$, $\mu_{\alpha}(i,j)=1$ for $i < i_{\alpha },j\ge j_{0}$. Therefore, the structure of the optimal policy for $j \ge j_0$ is determined, as illustrated in Fig.~\ref{fig:optproof}. Then, we examine the case for $j < j_0$. According to \eqref{L3a}, the optimal policy in each column is the same (0 or 2) for $i \ge i_0$. Thus, denote by
\begin{align}
j_{\alpha} = \max\{j: \mu_{\alpha}(i,j-1) = 0, \mu_{\alpha}(i,j) = \mu_{\alpha}(i,j+1) = \cdots = \mu_{\alpha}(i,j_0) = 2, \forall i\ge i_0\}.
\end{align}
By definition, we have $\mu_{\alpha}(i,j)=2$ for $i\ge i_{0}, j\ge j_{\alpha}$. In the following, we will show that $\mu_{\alpha}(i,j) = 1$ for all $i < i_{\alpha}, j \ge j_{\alpha}$, and $\mu_{\alpha}(i,j) = 2$ for all $i \ge i_{\alpha}, j \ge j_{\alpha}$. While for $j < j_{\alpha}$, the optimal policy is threshold-based between actions 0 and 1. Before that, denote by
	\begin{align}
		L_{\alpha}(i,j)&=J_{\alpha}(i+1,j)-J_{\alpha}(i,j),\label{Lij}\\
		D_{\alpha}(i)&=J_{\alpha}(i+1,i+1)-J_{\alpha}(i,i).\label{Di}
	\end{align}
We have the following lemma.

\begin{lemma} \label{lemma:LD}
For any $i \ge 1, j \ge j_0$, we have
	\begin{align}
		&D_{\alpha}(i+1)\le D_{\alpha}(i), \label{L5b}\\
		&L_{\alpha}(i,j)=L_{\alpha}(i,j+1)< D_{\alpha}(i), \label{L5a}
	\end{align}
\end{lemma}
\begin{proof}
See Appendix \ref{proof:LD}.
\end{proof}
Then, we can show the structure of the optimal policy in the whole state space $\mathcal{A}_{\mathrm{ext}}$.

\begin{lemma} \label{lemma:jalpha}
	Eqs. \eqref{L4a}, \eqref{L4c} and \eqref{L5a} hold for any $i \ge 1, j \ge j_{\alpha}$.
\end{lemma}
\begin{proof}
See Appendix \ref{proof:jalpha}.
\end{proof}

\begin{lemma} \label{lemma:ij}
	For any $j<j_{\alpha}$, there exist $I_{\alpha}(1)\le I_{\alpha}(2)\le \cdots\le I_{\alpha}(j_{\alpha}-1)\le i_{\alpha}$ so that
    \begin{align}
		\mu_{\alpha}(i,j)=\left\{\begin{array}{ll} 0,& i \ge I_{\alpha}(j),\\1,& i < I_{\alpha}(j).\end{array} \right. \label{L7a}
	\end{align}
\end{lemma}
\begin{proof}
See Appendix \ref{proof:ij}.
\end{proof}

Now, we can present the optimal policy for the average cost problem \eqref{prob:avgmin} as follows.

\begin{proposition} \label{thm:average}
There exists a threshold optimal stationary policy for the average cost problem \eqref{prob:avgmin}. In particular, with thresholds satisfying $\theta(1) \le \cdots \le \theta(\theta_{\mathrm r}-1) \le \theta_{\mathrm t} \le \theta_{\mathrm r}$, we have
\begin{align}
\mu(i,j) = \left\{\begin{array}{ll}
0, &\textrm{if~} i \ge \theta(j), j < \theta_{\mathrm{r}},\\
1, &\textrm{if~} i < \theta(j), j < \theta_{\mathrm{r}} \textrm{~or~} i < \theta_{\mathrm{t}}, j \ge \theta_{\mathrm{r}},\\
2, &\textrm{if~} i \ge \theta_{\mathrm{t}}, j \ge \theta_{\mathrm{r}}.
 \end{array} \right. \label{eq:mu}
\end{align}
\end{proposition}
\begin{proof}
See Appendix \ref{proof:average}.
\end{proof}

\section{Finding the Optimal Thresholds} \label{sec:solve}
We have shown that the optimal policy for the average cost problem \eqref{prob:avgmin} is threshold based. Thus, given the thresholds $\theta_{\mathrm t}, \theta_{\mathrm r}$,  and $\theta(j), 1\le j < \theta_{\mathrm r}$, a Markov chain can be obtained, as depicted in Fig.~\ref{fig:markov}. Denote by $\tilde{\mathcal{A}} = \{(i,i)| 1 \le i \le \theta_{\mathrm r}\} \cup \{(i,j)| 1 \le i \le \theta_{\mathrm t}, j = \theta_{\mathrm r}+ k \theta_{\mathrm t} +i, k \ge 0\}$. It contains all the states in the Markov chain. We have the following lemma.

\begin{figure}[t]
\centering
\includegraphics[width=5in]{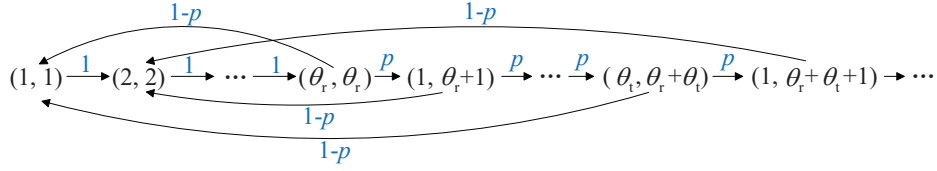}
\caption{Markov chain under the threshold optimal policy.} \label{fig:markov}
\end{figure}

\begin{lemma} \label{prop:transient}
Under the threshold policy \eqref{eq:mu}, all the states in $\mathcal{A} \setminus \tilde{\mathcal{A}}$ are transient.
\end{lemma}
\begin{proof}
With all the possible actions 0, 1, and 2, the state $(i,j) \in \mathcal{A} \setminus \tilde{\mathcal{A}}$ will transit to either $(1,1)$, $(i+1, i+1)$ or $(N, j+1)$, where $N = 1$ or $i+1$. Notice that to transit to the state $(i+1, i+1)$, the action 1 is taken, which means $i < \theta_{\mathrm t} \le \theta_{\mathrm r}$ according to the threshold policy. Thus, $(i+1, i+1) \in \tilde{\mathcal{A}}$. As $(1,1) \in \tilde{\mathcal{A}}$, if $(i,j)$ transients to either $(1,1)$ or $(i+1, i+1)$, the system state will transit within $\tilde{\mathcal{A}}$ afterwards. Otherwise, the AoIR will increase by 1. Repeating the process, the state transition ends up either within $\tilde{\mathcal{A}}$, or with the AoIR going to infinity. For both cases, $(i,j)$ will never be revisited.
\end{proof}

Based on Lemma \ref{prop:transient} and the fact that the stationary probabilities of transient states are all zero, to calculate the optimal average cost, we only need to calculate the stationary distribution for the states in $\tilde{\mathcal{A}}$. Thus, we have the following result.

\begin{theorem}\label{thm:closed}
The average AoIR and the average energy consumption for the problem \eqref{prob:avgmin} with the threshold policy \eqref{eq:mu} can be expressed as
\begin{align}
\bar{\Delta} &= \frac{\theta_{\mathrm t}}{2} + \frac{\theta_{\mathrm r}(\theta_{\mathrm r}-\theta_{\mathrm t})(1-p^{\theta_{\mathrm t}})}{2(\theta_{\mathrm r}(1-p^{\theta_{\mathrm t}})+\theta_{\mathrm t}p^{\theta_{\mathrm t}})} + \frac{1}{1-p}, \label{eq:aoiclosed}\\
\bar{E} &= \frac{1}{\theta_{\mathrm r}(1-p^{\theta_{\mathrm t}})+\theta_{\mathrm t}p^{\theta_{\mathrm t}}} \left(\frac{1-p^{\theta_{\mathrm t}}}{1-p} E_{\mathrm t} + E_{\mathrm s}\right). \label{eq:eclosed}
\end{align}
\end{theorem}
\begin{proof}
See Appendix \ref{proof:closed}.
\end{proof}

Theorem \ref{thm:closed} shows that the average cost is only relative to $\theta_{\mathrm t}$ and $\theta_{\mathrm r}$. The reason is that the state set related to the thresholds $\theta(j)$ can be expressed as $\hat{\mathcal{A}} = \{(i,j)| 1 \le i < \theta_{\mathrm t}, 1 \le j < \theta_{\mathrm r}, i < j\}$. Since $\hat{\mathcal{A}} \subset \mathcal{A} \setminus \tilde{\mathcal{A}}$, all the states in $\hat{\mathcal{A}}$ are transient, which do not contribute to the infinite horizon average cost. Thus, the values of $\theta(j)$ in the optimal policy can be arbitrarily set. For simplicity, a \emph{two-thresholds} optimal policy can be obtained by setting $\theta(1) = \cdots = \theta(\theta_{\mathrm r}-1) = 1$, as summarized in the following theorem.
{\begin{theorem} \label{thm:two}
There exists a two-thresholds optimal stationary policy for the average cost problem \eqref{prob:avgmin} as follows
\begin{align}
\mu(i,j) = \left\{\begin{array}{ll}
0, &\textrm{if~} j < \theta_{\mathrm{r}},\\
1, &\textrm{if~} i < \theta_{\mathrm{t}}, j \ge \theta_{\mathrm{r}},\\
2, &\textrm{if~} i \ge \theta_{\mathrm{t}}, j \ge \theta_{\mathrm{r}}.
 \end{array} \right. \label{eq:mutwo}
\end{align}
\end{theorem}

It turns out that the two-thresholds policy is quite intuitive and easy to be implemented. During transmission, it follows the truncated ARQ protocol~\cite{gong2018energy}. That is, a packet is retransmitted after a channel failure until either a successful transmission or the number of retransmissions reaches $\theta_{\mathrm t}-1$. In the latter case, a new packet is generated and is transmitted following the truncated ARQ protocol as well. The difference is that after a successful transmission, the sensor turns to sleep mode until the AoIR increases to $\theta_{\mathrm r}$. Then, another packet is generated and the truncated ARQ protocol restarts.
}

Furthermore, the optimal threshold $\theta_{\mathrm r}$ can be expressed in closed-form in terms of $\theta_{\mathrm t}$.
\begin{theorem} \label{coro:thetaropt}
The optimal threshold $\theta_{\mathrm r}$ for the problem \eqref{prob:avgmin} satisfies
\begin{align}
\theta_{\mathrm r} \in \{\max\{\theta_{\mathrm t}, \lceil \sqrt{A^2 + \theta_{\mathrm t} A + 2\omega B}-A \rceil\}, \max\{\theta_{\mathrm t}, \lfloor \sqrt{A^2 + \theta_{\mathrm t} A + 2\omega B}-A \rfloor \}\}, \label{eq:thetaropt}
\end{align}
where $A = \frac{\theta_{\mathrm t}p^{\theta_{\mathrm t}}}{1-p^{\theta_{\mathrm t}}}, B = \frac{1}{1-p}E_{\mathrm t} + \frac{1}{1-p^{\theta_{\mathrm t}}}E_{\mathrm s}$.
\end{theorem}
\begin{proof}
Denote $F(\theta_{\mathrm t}, \theta_{\mathrm r}) = \bar{\Delta} + \omega \bar{E}.$ We have
$$\frac{\partial F}{\partial \theta_{\mathrm r}} = \frac{\theta_{\mathrm r}^2 + 2A\theta_{\mathrm r} - \theta_{\mathrm t}A - 2\omega B}{2(\theta_{\mathrm r} + A)^2}.$$
The positive zero point is $\theta_{\mathrm r}^* = \sqrt{A^2 + \theta_{\mathrm t} A + 2\omega B}-A$, and $\frac{\partial F}{\partial \theta_{\mathrm r}} < 0$ for $0 < \theta_{\mathrm r} < \theta_{\mathrm r}^*$, $\frac{\partial F}{\partial \theta_{\mathrm r}} > 0$ for $\theta_{\mathrm r} > \theta_{\mathrm r}^*$. Therefore, to minimize $F$ with $\theta_{\mathrm r} \in \{1, 2, \cdots\}$ and $\theta_{\mathrm r} \ge \theta_{\mathrm t}$, we have \eqref{eq:thetaropt}.
\end{proof}

{The optimal thresholds are determined by the energy consumption model parameters $E_{\mathrm s}$ and $E_{\mathrm t}$, the channel parameter $p$, and the available energy condition implicitly indicated by $\omega$. Different sensors may have different thresholds, although the threshold structure is the same.} Notice that \cite[Lemma 2]{ceran2019average} is a special case of Theorem \ref{coro:thetaropt} when $E_{\mathrm s} = 0$. In this case, the action 2 is superior to the action 1 as the former possibly attains a lower AoIR without additional energy cost. Therefore, $\theta_{\mathrm t} = 1$, i.e., only the actions 0 and 2 are possibly taken.

Based on Theorems \ref{thm:closed} and \ref{coro:thetaropt}, the optimal thresholds $\theta_{\mathrm r}$ and $\theta_{\mathrm t}$ to minimize $\bar{\Delta} + \omega \bar{E}$ can be found by a line search over $\theta_{\mathrm t}$ and a comparison for the two values of $\theta_{\mathrm r}$ as in \eqref{eq:thetaropt}. Furthermore, when $\theta_{\mathrm t} \rightarrow \infty$, we have $p^{\theta_{\mathrm t}} \rightarrow 0$, and hence, $\sqrt{A^2 + \theta_{\mathrm t} A + 2\omega B}-A \rightarrow \sqrt{2\omega ((1-p)^{-1} E_{\mathrm t} + E_{\mathrm s})}$. Thus, there exists a constant $C > \sqrt{2\omega ((1-p)^{-1} E_{\mathrm t} + E_{\mathrm s})}$ such that when $\theta_{\mathrm t} \ge C$, we have $\theta_{\mathrm r} = \theta_{\mathrm t}$ and $\bar{\Delta} + \omega \bar{E} \approx \frac{1}{1-p} + \frac{\theta_{\mathrm t}}{2} + \omega \frac{(1-p)^{-1}E_{\mathrm t} + E_{\mathrm s}}{\theta_{\mathrm t}}$. As the average cost is non-decreasing over $\theta_{\mathrm t} \ge C$, the line search in the range $\theta_{\mathrm t} \le C$ is sufficient. In practice, when $p < 0.5$, $C \ge 50$ is sufficient for line search as $p^{50} < 10^{-16}$.

\section{Numerical Results} \label{sec:simulation}
In this section, numerical results are provided to verify our analytical conclusions. {Firstly, Fig.~\ref{fig:OptP} shows the optimal policies for the average cost problem \eqref{prob:avgmin} with different parameters. In Fig.~\ref{fig:OptP}(a), with $\omega = 2$, the optimal thresholds are $\theta_{\mathrm t} = 1, \theta_{\mathrm r} = 3$. With $\omega = 15$ as in Fig.~\ref{fig:OptP}(b), we have $\theta_{\mathrm t} = 3, \theta_{\mathrm r} = 8$. The reason is that the weighting factor $\omega$ indicates the relative importance of the energy consumption over the AoI. As $\omega$ increases, the optimal policy tends to be ``lazy" to save energy. Therefore, larger thresholds are adopted to have more chance to sleep. Comparing Fig.~\ref{fig:OptP}(b) with Fig.~\ref{fig:OptP}(c), we can see that as $E_{\mathrm t}$ increases, $\theta_{\mathrm t}$ decreases but $\theta_{\mathrm r}$ increases. The increase of $\theta_{\mathrm r}$ is because the sensor tends to sleep to save energy as the relative importance of energy increases. The decrease of $\theta_{\mathrm t}$ is because when the sensor should transmit, transmitting a new packet is more preferred than retransmission as the ratio between the transmit energy and the sensing energy increases.}

\begin{figure}
\centering
\includegraphics[width=6.0in]{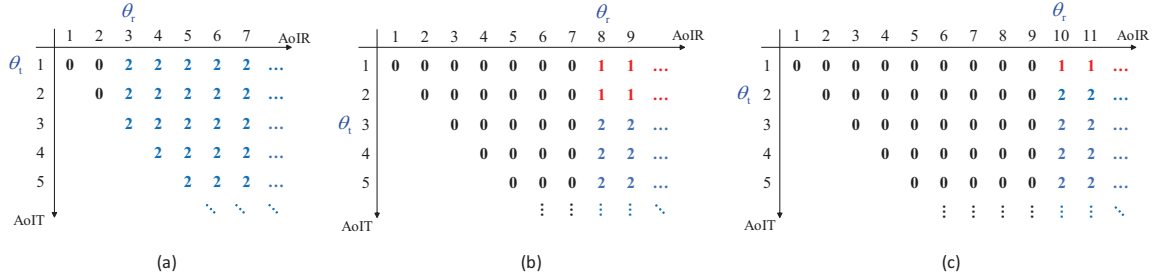}
\caption{The optimal policies for the average cost problem \eqref{prob:avgmin} with different parameters. (a) $p = 0.2, E_{\mathrm t} = E_{\mathrm s} = 1, \omega = 2$; (b) $p = 0.2, E_{\mathrm t} = E_{\mathrm s} = 1, \omega = 15$; (c) $p = 0.2, E_{\mathrm t} = 2, E_{\mathrm s} = 1, \omega = 15$.} \label{fig:OptP}
\end{figure}

{The minimum weighted sum cost performance is depicted in Fig.~\ref{fig:optcost}. Firstly, it is shown that the weighted sum cost is an increasing and concave function of the weighting factor. Secondly, the weighted sum cost increases as $p$ increases which causes the increase of transmission failures. In this case, both the average AoI and the average energy consumption may be larger due to more retransmissions. Finally, when $E_{\mathrm s}$ increases, the weighted sum cost increases as well because the total energy consumption increases.}

\begin{figure}
\centering
\includegraphics[width=3.4in]{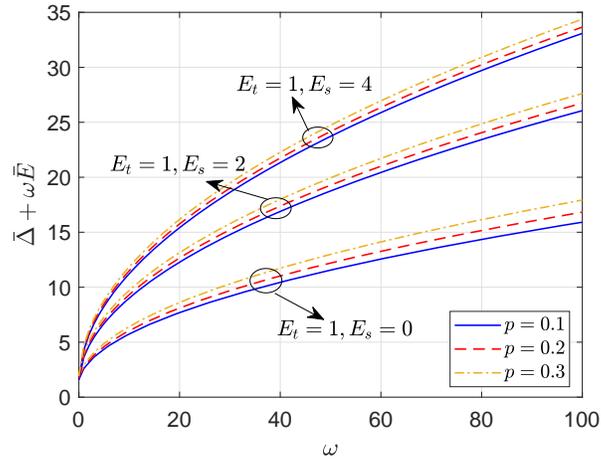}
\caption{The minimum weighted sum $\bar{\Delta} + \omega \bar{E}$ versus the weighting factor $\omega$.} \label{fig:optcost}
\end{figure}

The impacts of the channel parameter $p$ and the energy model parameters $E_{\mathrm s}$ and $E_{\mathrm t}$ on the energy-age tradeoff are shown in Figs.~\ref{fig:tradeoffp} and \ref{fig:tradeoffe}, respectively. In these figures, we compare our proposed two-thresholds policy with the single-threshold policy proposed in \cite{ceran2019average} and the truncated ARQ policy proposed in \cite{gong2018energy}. In the single-threshold policy, the sensor always senses and transmits a new packet when the AoIR exceeds the optimal threshold, and sleeps otherwise. Retransmission action is not included. In the truncated ARQ policy, the sensor transmits the same packet until either it is successfully received or the number of retransmissions reaches a threshold $M$. Then, a new packet is generated and transmitted. Thus, sleep mode is not considered.

{In Fig.~\ref{fig:tradeoffp}, we set $E_{\mathrm t} = E_{\mathrm s} = 1$.} With the increase of $p$, the tradeoff curve shifts towards upper right as the increase of channel failures causes both larger age and larger energy consumption. We also find that our policy performs better than the single-threshold policy. The reason is that the latter does not take retransmission action which can reduce energy cost when the sensing energy is large. In addition, the tradeoff range of our policy is far more wider than the truncated ARQ policy. This is because the truncated ARQ policy can only tradeoff age with energy by tuning $M$. While as sleep action is introduced in our policy, the tradeoff is more flexible.

\begin{figure}
\centering
\includegraphics[width=3.4in]{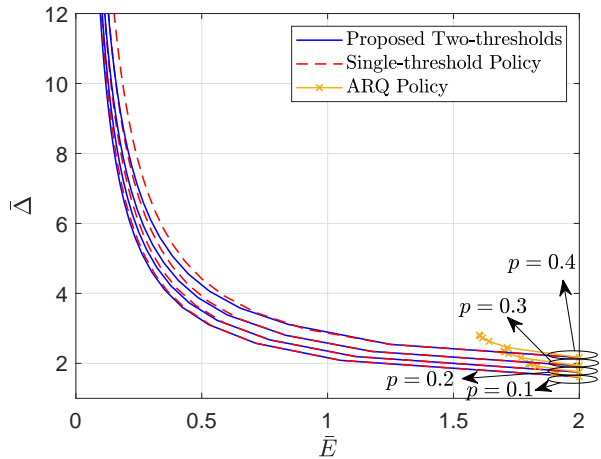}
\caption{The energy-age tradeoff curve for different values of $p$. $E_{\mathrm t} = E_{\mathrm s} = 1$.} \label{fig:tradeoffp}
\end{figure}

{In Fig.~\ref{fig:tradeoffe}, we set $p = 0.3$ and the total energy $E_{\mathrm t} + E_{\mathrm s} = 2$} to see the impact of sensing energy and transmit energy. It can be found that with the increase of the ratio $E_{\mathrm s}/E_{\mathrm t}$, the tradeoff curve is shifted towards left lower, meaning that our policy is more effective when the sensing energy dominates the transmit energy. The single-threshold policy with any ratio $E_{\mathrm s}/E_{\mathrm t}$ performs the same since the sensing action and the transmit action are bonded together. It performs the same as our policy with $E_{\mathrm t} = 0, E_{\mathrm s} = 2$, which validates the conclusion presented after Theorem \ref{coro:thetaropt}. Finally, similar to Fig.~\ref{fig:tradeoffp}, the truncated ARQ policy results in a narrow tradeoff range.

\begin{figure}
\centering
\includegraphics[width=3.4in]{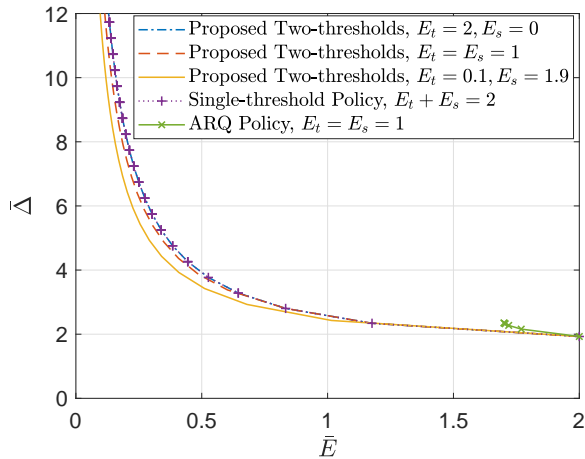}
\caption{The energy-age tradeoff curve for different energy models. $p = 0.3$.} \label{fig:tradeoffe}
\end{figure}

\section{Conclusion} \label{sec:conclusion}
In this paper, the tradeoff between AoI and energy consumption over an error-prone channel considering sleep and retransmission has been studied. We proved that there exists a threshold optimal stationary policy for the average cost minimization problem. In addition, the minimum cost only rely on two thresholds $\theta_{\mathrm t}$ and $\theta_{\mathrm r}$, and $\theta_{\mathrm r}$ can be expressed in closed-form in terms of $\theta_{\mathrm t}$. Thus, the optimal policy can be easily obtained by a line search over $\theta_{\mathrm t}$, and a two-thresholds optimal policy can be obtained. Then, we numerically illustrate the tradeoff between the average AoI and the average energy consumption. Compared with the truncated ARQ protocol, our optimal policy attains a tradeoff curve towards left lower, and the tradeoff range is substantially expanded. { Future work can exploit the practical channel correlations by introducing channel training, and consider transmit power control and data rate adaptation.}

\appendix

\subsection{Proof of Proposition \ref{prop:discBellman}} \label{proof:discBellman}

Due to the optimality of the DP algorithm for a finite horizon problem, it can be verified by induction that $J_{n}(x)$ is the optimal cost for the $n$-stage discounted problem with initial state $x$, per-stage cost $g$ and terminal cost 0, i.e.,
\begin{align*}
J_{n}(x) = \min_{\pi_n} J_{\pi_n}^{(n)}(x) = \min_{\pi_n} \mathbb{E}\left( \sum_{k=0}^{n-1} \alpha^k g(x_k, \mu_k(x_k))\Bigg| x_0 = x\right),
\end{align*}
where $\pi_n = \{\mu_0, \cdots, \mu_{n-1}\}$.

Secondly, we follow the lines of proof of \cite[Prop.~1.2.1]{bertsekas2005dynamic} to show the convergence of $J_{n}(x)$. For every positive integer $n$, initial state $x_0 = (i,j) \in \mathcal{A}$ and policy $\pi = \{\mu_0, \mu_1, \cdots\}$, consider the ``tail" cost after stage $n$, we have
\begin{align}
\left| J_{\alpha,\pi}(x_0) - J_{\pi_n}^{(n)}(x_0) \right| &= \left|\lim_{K\rightarrow \infty} \mathbb{E}\left( \sum_{k=n}^{K-1} \alpha^k g(x_k, \mu_k(x_k)) \right)\right| \nonumber\\
&\le \lim_{K\rightarrow \infty} \sum_{k=n}^{K-1} \alpha^k \left(j + k + \omega(E_{\mathrm s} + E_{\mathrm t})\right) = \left(M(x_0) + n\right) \frac{\alpha^n}{1-\alpha}, \label{eq:tail}
\end{align}
where $\pi_n$ equals to the first $n$ stage policy of $\pi$, $M(x_0) = j + \omega(E_{\mathrm s} + E_{\mathrm t})+\frac{\alpha}{1-\alpha}$, and the inequality holds as the AoIR grows at most linearly versus $k$. Therefore, we have
\begin{align*}
J_{\alpha,\pi}(x_0) - \left(M(x_0) + n\right) \frac{\alpha^n}{1-\alpha} \le J_{\pi_n}^{(n)}(x_0) \le J_{\alpha,\pi}(x_0) + \left(M(x_0) + n\right) \frac{\alpha^n}{1-\alpha}.
\end{align*}
By taking the minimum over $\pi_n$ and $\pi$, we have
\begin{align*}
J_{\alpha}(x_0) - \left(M(x_0) + n\right) \frac{\alpha^n}{1-\alpha} \le J_{n}(x_0) \le J_{\alpha}(x_0) + \left(M(x_0) + n\right) \frac{\alpha^n}{1-\alpha}.
\end{align*}
By taking $n \rightarrow \infty$, $\left(M(x_0) + n\right) \frac{\alpha^n}{1-\alpha} \rightarrow 0$. Therefore, we have $J_{\alpha}(x_0) \le \lim_{n \rightarrow \infty}J_{n}(x_0) \le J_{\alpha}(x_0)$, which means
$J_{\alpha}(x_0) = \lim_{n \rightarrow \infty} J_{n}(x_0).$

Then, following the proofs of \cite[Prop.~1.2.2]{bertsekas2005dynamic} and \cite[Prop.~1.2.3]{bertsekas2005dynamic}, we can show that the optimal cost $J_{\alpha}(x_0)$ satisfies the Bellman's Equation \eqref{eq:bellman}, and the optimal solution $\mu_{\alpha}$ for \eqref{eq:bellman} forms the optimal stationary policy.

\subsection{Proof of Lemma \ref{lemma:monoton}} \label{proof:monoton}
Firstly, we prove \eqref{eq:nmonoj} by induction.

(a) It can be easily obtained that $J_{1}(i,j) = j, \forall (i,j) \in \mathcal{A}$. Therefore, \eqref{eq:nmonoj} holds for $n = 1$.

(b) Suppose $J_{n}(i, j+1) - J_{n}(i, j) \ge 1$ holds for $n \le k$. Thus, $J_{k}(i, j+1) \ge J_{k}(i, j)$ for any $1 \le i \le j$. We consider $n = k+1$. Denote
$u_{k+1}^* = \arg \min_{u\in\mathcal{U}} \left[g(x, u) + \alpha \sum_{x'\in \mathcal{A}}p_{x\rightarrow x'}(u)J_{k}(x')\right]$.
If $u_{k+1}^* = (0,0)$, we have
\begin{align*}
J_{k+1}(i, j+1) = j+1 + \alpha J_{k}(i + 1, j + 2) \ge j + 1 +  \alpha J_{k}(i + 1, j + 1) \ge 1+ J_{k+1}(i, j),
\end{align*}
where the second inequality is due to the optimality of $J_{k+1}(i, j)$ over all possible actions. Similarly, we can prove that the inequality $J_{k+1}(i, j+1) - J_{k+1}(i, j) \ge 1$ holds for any $u_{k+1}^* \in \mathcal{U}$. Therefore, \eqref{eq:nmonoj} holds for any $n \ge 1$. 

The inequality \eqref{eq:nmonoi} can be proved in the same way. This completes the proof.

\subsection{Proof of Proposition \ref{prop:converge}} \label{proof:converge}
The convergence of $\mu_{\beta_m}$ and the existence of the stationary policy $\mu$ is given by \cite[Lemma]{sennott1989average}. According to \cite[Theorem]{sennott1989average}, to prove that $\mu$ is the optimal policy for the average cost minimization problem, the following conditions should be satisfied:

1) $J_{\alpha}(x)$ is finite for every $x$ and $\alpha$.

2) There exists a nonnegative $N$ such that $h_{\alpha}(x) = J_{\alpha}(x) - J_{\alpha}(\hat x) \ge -N$ for all $x$ and $\alpha$, where $\hat x$ is any fixed state.

3) There exists nonnegative $M_x$, such that $h_{\alpha}(x) \le M_x$ for every $x$ and $\alpha$. Moreover, for every $x$, there exists an action $u$ such that $\sum_{x'}p_{x \rightarrow x'}(u)M_{x'} < \infty$.

We verify the above conditions one by one.

\emph{Validity of condition 1):} Denote by $J_{\alpha, \pi}(x)$ the discounted cost under the policy $\pi$. Due to the optimality of $J_{\alpha}(x)$, we have
$J_{\alpha}(x) \le J_{\alpha, \pi}(x) < \infty$
for any $x$ and $\alpha$ according to \eqref{eq:finite}.

\emph{Validity of condition 2):} By setting $\hat x = (1,1)$ and $N = 0$, according to Lemma \ref{lemma:monoton}, we have $J_{\alpha}(x) \ge J_{\alpha}(\hat x), \forall x \in \mathcal{A}$. Therefore, we have $h_{\alpha}(x) = J_{\alpha}(x) - J_{\alpha}(\hat x) \ge 0.$

\emph{Validity of condition 3):} We set $\hat x = (1,1)$, define $K_{\min} = \min\{k:k\ge 1, x_k = (1,1)\}$, and consider the policy $\hat{\pi} = \{\hat{\mu}_0, \hat{\mu}_1, \cdots\}$ that always adopts the action $(1,1)$ until the $K_{\min}$-th slot, and then is the same as the optimal stationary policy $\mu$. Denote $G(x, \hat x)$ as the expected sum cost of the first passage from any state $x = (i, j)$ to $\hat x$ under policy $\hat \pi$. Notice that $\mathrm{Pr}(K_{\min} = k) = (1-p)p^{k-1}$, and the cost is fixed for a given $K_{\min}$. Then, $G(x, \hat x)$ can be computed as
\begin{align*}
G(x, \hat x) &=  \mathbb{E}\left( \sum_{k=0}^{K_{\min}\!-\!1} g(x_k, \hat{\mu}_k)\bigg|x_0 = x, \hat{\pi}\right) = \sum_{k=1}^{\infty} (1-p)p^{k-1} \left(\sum_{m=0}^{k-1} \left(j + m + \omega(E_{\mathrm s} + E_{\mathrm t}) \right)\right)\nonumber\\
&= \frac{j - \frac{1}{2} + \omega(E_{\mathrm s} + E_{\mathrm t})}{1-p} + \frac{1+p}{2(1-p)^2},
\end{align*}
Thus, we have
\begin{align*}
J_{\alpha}(x) \le &{~}\mathbb{E}\!\left( \sum_{k=0}^{K_{\min}-1} \alpha^kg(x_k, \hat{\mu}_k(x_k))\bigg|x_0 \!=\! x, \hat{\pi}\right) \!+\! \lim_{K\rightarrow \infty}\mathbb{E}\!\left( \sum_{k = K_{\min}}^K \alpha^kg(x_k, \hat{\mu}_k(x_k))\bigg|x_{K_{\min}} \!=\! \hat x, \hat{\pi}\right)\nonumber\\
\le &{~}\mathbb{E}\left( \sum_{k=0}^{K_{\min}-1} g(x_k, \hat{\mu}_k(x_k))\bigg|x_0 \!=\! x, \hat{\pi}\right) \!+\! \lim_{K\rightarrow \infty}\mathbb{E}\left( \sum_{k = 0}^K \alpha^kg(x_k, {\mu}(x_k))\bigg|x_{0} \!=\! \hat x, {\mu}\right) \nonumber\\
=&{~}G(x, \hat x) \!+\! J_{\alpha}(\hat x).
\end{align*}

By setting $M_x = G(x, \hat x)$, we have $h_{\alpha}(x) = J_{\alpha}(x) - J_{\alpha}(\hat x) \le G(x, \hat x) = M_x$ for every $x$ and $\alpha$. In addition, by taking $u = (0,0)$ in state $x = (i, j)$, we have $\sum_{x'}p_{x \rightarrow x'}(u)M_{x'} = M_{(i+1, j+1)} < \infty$.

In summary, all the assumptions in \cite[Theorem]{sennott1989average} hold. This completes the proof.

\subsection{Proof of Lemma \ref{lemma:action}} \label{proof:action}
We prove it by contradiction. Suppose for a certain state $x = (i, j)$ we have $\mu_{\alpha}(x) = (1,0)$. Without loss of generality, we set $x_0 = x$ and the optimal action is $u_0 = (1,0)$. Denote $N = \min\{k: k > 0, u_k \neq (0,0)\}$ as the stage where the first non-idle action is taken under the optimal policy. Then, $u_N$ has three possible choices.

\emph{Case 1:} If $u_N = (1, 0)$, consider the policy $\pi'$ in which the action $(0,0)$ is taken in stage 0, and the optimal policy is taken afterwards, we have
\begin{align*}
J_{\alpha}(i, j) &= (j + \omega E_{\mathrm s}) + \sum_{k=1}^{N-1} \alpha^k (j+k) + \alpha^N(j+N + \omega E_{\mathrm s}) + \alpha^{N+1} J_{\alpha}(1, j + N+1) \nonumber\\
&> \sum_{k=0}^{N-1} \alpha^k (j+k) + \alpha^N(j+N + \omega E_{\mathrm s}) + \alpha^{N+1} J_{\alpha}(1, j + N+1) = J_{\alpha, \pi'}(i, j).
\end{align*}

\emph{Case 2:} If $u_N = (0, 1)$, consider the policy $\pi'$ in which the action $(0,0)$ is taken in stages $k = 0, 1, \cdots, N-1$, the action $(1,1)$ is taken in stage $N$, and the optimal policy is taken afterwards, we have
\begin{align*}
J_{\alpha}(i, j) &= (j + \omega E_{\mathrm s}) + \sum_{k=1}^{N-1} \alpha^k (j+k) + \alpha^N(j+N + \omega E_{\mathrm t})  \nonumber\\
&\qquad\qquad +\alpha^{N+1} \left(pJ_{\alpha}(N+1, j + N+1) + (1-p)J_{\alpha}(N+1, N+1)\right).  \nonumber\\
&> \sum_{k=0}^{N-1} \alpha^k (j\!+\!k) \!+\! \alpha^N(j\!+\!N \!+ \!\omega E_{\mathrm s} \!+ \!\omega E_{\mathrm t}) \!+\!\alpha^{N+1} \left(pJ_{\alpha}(1, j \!+\! N\!+\!1) \!+\!(1\!-\!p)J_{\alpha}(1, 1) \right) \nonumber\\
&= J_{\alpha, \pi'}(i, j)
\end{align*}
as $J_{\alpha}(i,j)$ is nondecreasing versus $i$ and $j$ according to Lemma \ref{lemma:monoton}.

\emph{Case 3:} If $u_N = (1, 1)$, consider the same policy $\pi'$ as in case 2, we also have $J_{\alpha, \pi'}(i, j)
< J_{\alpha}(i, j)$ as the policy $\pi'$ reduces the cost in stage $k=0$ without impacting the future costs.

In summary, we end up with a contradiction that $\mu_{\alpha}(x) = (1,0)$ is not optimal. Therefore, the optimal policy must not contain the action $(1,0)$.

\subsection{Proof of Lemma \ref{lemma:i0}} \label{proof:i0}
We prove the lemma by induction.

(a) For $n=1$, \eqref{L3a0} and \eqref{L3b0} hold as $\mu_1(i,j) = 0$ and $J_{1}(i,j)=j$ for all $i \ge 1, j \ge 1$.

(b) Suppose \eqref{L3a0} and \eqref{L3b0} hold for $n\le k$, we consider $n=k+1$. Firstly, we prove $\mu_{k+1}(i,j)\in \{0,2\}$. Denote by
\begin{align}
	A_{n}^{u-v}(i,j)&=Q_{n}(i,j,u)-Q_{n}(i,j,v), \quad n \ge 1, u, v \in \{0, 1, 2\}. \label{Auv}
\end{align}
If $i \ge i_0$, according to Lemma \ref{lemma:monoton}, we have
	\begin{align*}
		A_{k+1}^{1-2}(i,j)
		&=\alpha p\left( J_{k}(i\!+\!1,j\!+\!1)-J_{k}(1,j\!+\!1)\right) +\alpha \left(1-p\right)\left( J_{k}(i\!+\!1,i\!+\!1)-J_{k}(1,1)\right) -\omega E_{\mathrm{s}} \nonumber\\
		&\ge \alpha \left(1-p \right) \left(  J_{k}(1,i+1)-J_{k}(1,1)\right) - \omega E_{\mathrm{s}}> \alpha \left(1-p \right) i-\omega E_{\mathrm{s}} > 0,
	\end{align*}
which means $Q_{k+1}(i,j,1)>Q_{k+1}(i,j,2)$. Hence, $\mu_{k+1}(i,j) \in \{0, 2\}$. Similarly, if $j \le i$,
	\begin{align*}
		A_{k+1}^{0-1}(i,j)
		&=\alpha(1-p)(J_{k}(i+1,j+1)-J_{k}(i+1,i+1))-\omega E_{\mathrm t} \le -\omega E_{\mathrm t} < 0.
	\end{align*}
Thus, we also have $\mu_{k+1}(i,j)\in \{0,2\}$ if $i \ge j$.

	Secondly, we prove $\mu_{k+1}(i,j)=\mu_{k+1}(i+1,j)$. As \eqref{L3b0} holds for $n=k$, we have
	\begin{align*}
		A_{k+1}^{0-2}(i,j)&=\alpha J_{k}(i+1,j+1) - \alpha p J_{k}(1,j+1) - \alpha \left(1-p\right)J_{k}(1,1)-\omega E_{\mathrm{t}}-\omega E_{\mathrm{s}} \nonumber\\
		&=\alpha J_{k}(i+2,j+1) - \alpha p J_{k}(1,j+1) - \alpha \left(1-p\right)J_{k}(1,1)-\omega E_{\mathrm{t}}-\omega E_{\mathrm{s}}\nonumber\\
&=A_{k+1}^{0-2}(i+1,j).
	\end{align*}
	Therefore, we have $\mu_{k+1}(i,j)=\mu_{k+1}(i+1,j)$. Thus, \eqref{L3a0} holds for $n=k+1$.

    Next, we prove \eqref{L3b0}. As \eqref{L3a0} holds for $n=k+1$, there are two possible cases.

    If $\mu_{k+1}(i,j)=\mu_{k+1}(i+1,j)=2$, we have
    \begin{align*}
    J_{k+1}(i,j)= j + \omega E_{\mathrm{t}} + \omega E_{\mathrm{s}} + \alpha p J_k(1,j+1) + \alpha (1-p) J_k(1,1) = J_{k+1}(i+1,j).
    \end{align*}
	
    If $\mu_{k+1}(i,j)=\mu_{k+1}(i+1,j)=0$, we have
    \begin{align*}
    J_{k+1}(i,j)-J_{k+1}(i+1,j)=\alpha\left( J_{k}(i+1,j+1)-J_{k}(i+2,j+1)\right) =0.
    \end{align*}
	Therefore, \eqref{L3b0} holds for $n=k+1$. This completes the proof.

\subsection{Proof of Lemma \ref{lemma:j0}} \label{proof:j0}
We prove the lemma by induction.

(a) For $n=2$, $\mu_2(i,j)$ can be directly calculated as
\begin{align*}
\mu_2(i,j) = \left\{\begin{array}{ll}0, & \textrm{if~} i <i_0,j < \Big\lceil \dfrac{\omega E_{\mathrm t}}{\alpha(1-p)}\Big\rceil+i \textrm{~or~} i \ge i_0, j < j_0, \\
1, & \textrm{if~} i <i_0,j \ge \Big\lceil \dfrac{\omega E_{\mathrm t}}{\alpha(1-p)}\Big\rceil+i,\\
2, & \textrm{if~}i \ge i_0, j \ge j_0.
\end{array} \right.
\end{align*}
Then, it can be easily obtained that for $j \ge j_0$, $J_{2}(i,j+1)-J_{2}(i,j) = 1 + \alpha p$. Hence, \eqref{L4a0} and \eqref{L4c0} hold for $n = 2, j \ge j_0$ with $i_2 = i_0$.

(b) Suppose \eqref{L4a0} and \eqref{L4c0} hold for $n\le k$, we consider $n\!=\!k\!+\!1$. Based on Lemma \ref{lemma:monoton}, we have
\begin{align*}
	A_{k\!+\!1}^{0\!-\!2}(i,j)&=\alpha p(J_{k}(i\!+\!1,j\!+\!1)\!-\!J_{k}(1,j\!+\!1))\!+\! \alpha(1\!-\!p)(J_{k}(i\!+\!1,j\!+\!1)\!-\!J_{k}(1,1))\!-\!\omega E_{\mathrm t}\!-\!\omega E_{\mathrm s}\nonumber\\
	&\ge \alpha(1\!-\!p)(J_{k}(1,j\!+\!1)\!-\!J_{k}(1,1))\!-\!\omega E_{\mathrm t}\!-\!\omega E_{\mathrm s} \ge \alpha(1\!-\!p)j\!-\!\omega E_{\mathrm t}\!-\!\omega E_{\mathrm s}>0
\end{align*}
if $j \ge j_0$, where $A_{k+1}^{0-2}(i,j)$ is defined in \eqref{Auv}. Thus, we have $Q_{n}(i,j,0)>Q_{n}(i,j,2)$, which means $\mu_{k+1}(i,j)\in \{1,2\} $. Then, as \eqref{L4c0} holds for $n=k$, we have
\begin{align*}
	A_{k\!+\!1}^{1\!-\!2}(i,j)\!-\!A_{k\!+\!1}^{1\!-\!2}(i,j\!+\!1) = \alpha p \left(J_{k}(i\!+\!1,j\!+\!1)\!-\!J_{k}(i\!+\!1,j\!+\!2)\!-\!\left(J_{k}(1,j\!+\!1)\!-\!J_{k}(1,j\!+\!2)\right)\right)=0.
\end{align*}
Therefore, $\mu_{k+1}(i,j)=\mu_{k+1}(i,j+1)$ for $j \ge j_0$. As
\begin{align*}
A_{k+1}^{1-2}(i,j)=\alpha p(J_{k}(i\!+\!1,j\!+\!1)\!-\!J_{k}(1,j\!+\!1))\!+\!\alpha(1\!-\!p)(J_{k}(i\!+\!1,i\!+\!1)\!-\!J_{k}(1,1)) \!-\!\omega E_{\mathrm s}
\end{align*}
is strictly increasing versus $i$, there exists $i_{k+1}$ so that $A_{k+1}^{1-2}(i,j)<0$ for $i<i_{k+1}$ and $A_{k+1}^{1-2}(i,j)>0$ for $i\ge i_{k+1}$. Therefore, \eqref{L4a0} holds for $n = k+1$.

As \eqref{L4a0} holds for $n = k+1$ and \eqref{L4c0} holds for $n = k$, we have
\begin{align*}
	J_{k+1}(i,j+1)-J_{k+1}(i,j)=1+\alpha p \left(J_{k}(N,j+2)-J_{k}(N,j+1)\right)=\dfrac{1-\left(\alpha p\right)^{k+1} }{1-\alpha p},
\end{align*}
where $N=1$ if $\mu_{k+1}(i,j)=\mu_{k+1}(i,j+1)=2$ and $N=i+1$ if $\mu_{k+1}(i,j)=\mu_{k+1}(i,j+1)=1$. Therefore, \eqref{L4c0} holds for $n=k+1$. This completes the proof.

\subsection{Proof of Lemma \ref{lemma:LD}} \label{proof:LD}
Firstly, we prove by induction that
	\begin{align}
		J_{n}(i,j+2)-J_{n}(i,j+1)\le J_{n}(i,j+1)-J_{n}(i,j).\label{L2a}
	\end{align}

(a)	As $J_{1}(i,j)=j$, \eqref{L2a} holds for $n=1$.

(b) Suppose \eqref{L2a} holds for $n\le k$, we consider $n = k+1$. If $\mu_{k+1}(i,j+1) = 0$, we have $J_{k+1}(i,j+1) = Q_{k+1}(i,j+1,0)$. Due to the optimality of $J_{k+1}(i,j)$, we have
\begin{align*}
	&{~}J_{k+1}(i,j+2)-J_{k+1}(i,j+1)-\left(  J_{k+1}(i,j+1)-J_{k+1}(i,j)\right) \nonumber\\
	=&{~}J_{k+1}(i,j+2)-Q_{k+1}(i,j+1,0)-\left(  Q_{k+1}(i,j+1,0)-J_{k+1}(i,j)\right) \nonumber\\
	\le&{~} Q_{k+1}(i,j+2,0)-Q_{k+1}(i,j+1,0)-\left(  Q_{k+1}(i,j+1,0)-Q_{k+1}(i,j,0)\right) \nonumber\\
	=&{~} \alpha \left(J_{k}(i+1,j+3)-J_{k}(i+1,j+2)-\left( J_{k}(i+1,j+2)-J_{k}(i+1,j+1)\right) \right) \le 0.
\end{align*}

If $\mu_{k+1}(i,j+1) = 1$ or $\mu_{k+1}(i,j+1) = 2$, the inequality can be proved in the same way. Therefore, \eqref{L2a} holds for all $n \ge 1$. As $n \rightarrow \infty$, we have
	\begin{align}
		J_{\alpha}(i,j+2)-J_{\alpha}(i,j+1)\le J_{\alpha}(i,j+1)-J_{\alpha}(i,j). \label{eq:concave}
	\end{align}

According to Corollary \ref{coro:i0j0}, we have $J_{\alpha}(i,j)=J_{\alpha}(i+1,j)$ for any $j\le i$. Therefore,
\begin{align*}
	D_{\alpha}(i+1)&=J_{\alpha}(i+2,i+2)-J_{\alpha}(i+1,i+1)=J_{\alpha}(i+2,i+2)-J_{\alpha}(i+2,i+1) \nonumber\\
	&\le J_{\alpha}(i+2,i+1)-J_{\alpha}(i+2,i)=J_{\alpha}(i+1,i+1)-J_{\alpha}(i,i)=D_{\alpha}(i),
\end{align*}
where the inequality holds according to \eqref{eq:concave}. Hence, \eqref{L5b} is proved.

According to \eqref{L4c}, we have
\begin{align*}
	L_{\alpha}(i,j)-L_{\alpha}(i,j+1)&=J_{\alpha}(i+1,j)-J_{\alpha}(i,j)-\left(J_{\alpha}(i+1,j+1)-J_{\alpha}(i,j+1)\right) \nonumber\\
	&=J_{\alpha}(i+1,j)-J_{\alpha}(i+1,j+1)-\left(J_{\alpha}(i,j)-J_{\alpha}(i,j+1)\right)=0.
\end{align*}
Thus, we have $L_{\alpha}(i,j)=L_{\alpha}(i,j+1)$.

Finally, we prove that $L_{\alpha}(i,j)<D_{\alpha}(i)$. Based on Lemma \ref{lemma:monoton}, we have $D_{\alpha}(i)\ge 1$. Thus, for $i\ge i_{\alpha}, j\ge j_{0}$, as $\mu_{\alpha}(i,j) = 2$, we have
\begin{align*}
L_{\alpha}(i,j)=Q_{\alpha}(i+1,j,2)-Q_{\alpha}(i,j,2)=0<D_{\alpha}(i).
\end{align*}

For $i \le i_{\alpha}, j\ge j_{0}$, we prove the inequality $L_{\alpha}(i,j)< D_{\alpha}(i)$ by induction.

(a) The inequality holds for $i = i_{\alpha}$.

(b) Suppose the inequality holds for $i\ge k$, where $1 < k \le i_{\alpha}$, and consider $i=k-1$. As $\mu_{\alpha}(k-1,j) = 1$ according to \eqref{L4a}, we have
\begin{align*}
	L_{\alpha}(k-1,j)&\le Q_{\alpha}(k,j,1)-Q_{\alpha}(k-1,j,1)=\alpha p L_{\alpha}(k,j+1)+\alpha \left(1-p\right)D_{\alpha}(k) \nonumber\\
	&<\alpha p D_{\alpha}(k)+\alpha \left(1-p\right)D_{\alpha}(k)=\alpha D_{\alpha}(k)<D_{\alpha}(k-1).
\end{align*}
Therefore, the inequality holds for  $i=k-1$. This completes the proof.

\subsection{Proof of Lemma \ref{lemma:jalpha}} \label{proof:jalpha}
Based on Lemma \ref{lemma:j0} and Lemma \ref{lemma:LD}, \eqref{L4a}, \eqref{L4c} and \eqref{L5a} hold for $j \ge j_0$. 

For $j_{\alpha} \le j \le j_0$, we prove by induction. Firstly, \eqref{L4a}, \eqref{L4c} and \eqref{L5a} hold for $j=j_{0}$. Secondly, suppose the equations hold for $j \ge k$, where $j_{\alpha} < k \le j_{0}$, and consider $j=k-1$. The proof of \eqref{L4a} is divided into the following four steps.

(a) For $i \ge i_0$, by definition, we have $\mu_{\alpha}(i,k-1) = 2$. Hence, $J_{\alpha}(i,k)-J_{\alpha}(i,k-1) = 1+\alpha p (J_{\alpha}(1,k+1)-J_{\alpha}(1,k)) =\frac{1}{1-\alpha p}$. In addition, $L_{\alpha}(i,k-1)=L_{\alpha}(i,k)=0<D_{\alpha}(i)$. Therefore, the lemma holds for $j = k-1$ when $i \ge i_0$.

(b) Denote by $A_{\alpha}^{u-v}(i,j) = \lim_{n \rightarrow \infty} A_{n}^{u-v}(i,j).$ For $i_{\alpha} \le i \le i_0$, since $\mu_{\alpha}(i,k)=2$, we have $A_{\alpha}^{0-2}(i,k)>0$ and $J_{\alpha}(i_{\alpha},k)=J_{\alpha}(i_{\alpha}+1,k)=\cdots=J_{\alpha}(i_{0},k)$. As
\begin{align}
A_{\alpha}^{0-2}(i,j)=\alpha J_{\alpha}(i + 1,j + 1)  -  \alpha p J_{\alpha}(1,j + 1)  -  \alpha \left(1 - p\right)J_{\alpha}(1,1) - \omega E_{\mathrm{t}} - \omega E_{\mathrm{s}}, \label{eq:A02}
\end{align}
we have $A_{\alpha}^{0-2}(i_{\alpha}-1,k-1)=A_{\alpha}^{0-2}(i_{\alpha},k-1)=\cdots=A_{\alpha}^{0-2}(i_{0},k-1)>0$. Therefore, $\mu_{\alpha}(i,k-1)\in \{1,2\}$ for $i_{\alpha}-1 \le i \le i_{0}$.

(c) As \eqref{L5a} holds for $j\ge k$, we have
\begin{align}
	A_{\alpha}^{1-2}(i,k-1)&=\alpha p(J_{\alpha}(i+1, k) - J_{\alpha}(1,k)) + \alpha (1-p) (J_{\alpha}(i+1,i+1) - J_{\alpha}(1,1)) - \omega E_{\mathrm{s}} \nonumber\\
    &=\alpha p\sum_{m=1}^{i}L_{\alpha}(m,k)+\alpha \left(1-p\right)\sum_{m=1}^{i} D_{\alpha}(m)-\omega E_{\mathrm{s}}\nonumber\\
	&=\alpha p\sum_{m=1}^{i}L_{\alpha}(m,k+1)+\alpha \left(1-p\right)\sum_{m=1}^{i} D_{\alpha}(m)-\omega E_{\mathrm{s}}=A_{\alpha}^{1-2}(i,k). \label{eq:A12l}
\end{align}
As \eqref{L4a} holds for $j=k$, we have $A_{\alpha}^{1-2}(i,k)>0$ for $i\ge i_{\alpha}$. Thus, $A_{\alpha}^{1-2}(i,k-1)>0$ and hence, $\mu_{\alpha}(i,k-1)\in \{0,2\}$ for $i\ge i_{\alpha}$. Joint with (b), we have $\mu_{\alpha}(i,k-1)=2$ for $i\ge i_{\alpha}$.

(d) For $i < i_{\alpha}$, we prove that $\mu_{\alpha}(i,k-1) = 1$ by contradiction. Suppose $\mu_{\alpha}(i,k-1)=1$ and $\mu_{\alpha}(i-1,k-1)=0$. On one hand, we have
\begin{align*}
	L_{\alpha}(i-1,k-1) \le Q_{\alpha}(i,k-1,0)-Q_{\alpha}(i-1,k-1,0)=\alpha L_{\alpha}(i,k).
\end{align*}
On the other hand, as $L_{\alpha}(i,k)<D_{\alpha}(i)$,
\begin{align*}
	L_{\alpha}(i-1,k-1)&\ge Q_{\alpha}(i,k-1,1)-Q_{\alpha}(i-1,k-1,1)=\alpha p L_{\alpha}(i,k)+\alpha \left(1-p\right)D_{\alpha}(i) \nonumber\\
	&>\alpha p L_{\alpha}(i,k)+\alpha \left(1-p\right)L_{\alpha}(i,k)=\alpha L_{\alpha}(i,k).
\end{align*}
Thus, we have $\alpha L_{\alpha}(i,k)<\alpha L_{\alpha}(i,k)$, which is a contradiction. Therefore, if $\mu_{\alpha}(i,k-1)=1$, we must have $\mu_{\alpha}(i-1,k-1)\neq 0$. According to (c), we have $A_{\alpha}^{1-2}(i,k-1)<0$ for $i< i_{\alpha}$. Hence, $\mu_{\alpha}(i,k-1)\in \{0,1\}$. Joint with (b), we have $\mu_{\alpha}(i_{\alpha}-1,k-1)=1$. Hence, we have $\mu_{\alpha}(i,k-1)=1$ for any $i<i_{\alpha}$. Thus, \eqref{L4a} holds for $j = k-1$.

As $\mu_{\alpha}(i,k-1)=\mu_{\alpha}(i,k) \in \{1,2\}$, we have
\begin{align*}
	J_{\alpha}(i,k)-J_{\alpha}(i,k-1)=1+\alpha p \left(J_{\alpha}(N,k+1)-J_{\alpha}(N,k)\right)=\frac{1}{1-\alpha p},
\end{align*}
where $N=1$ if $\mu_{\alpha}(i,k)=2$ and $N=i+1$ if $\mu_{\alpha}(i,k)=1$. Thus, \eqref{L4c} holds for $j=k-1$.

Similar to the proof of Lemma \ref{lemma:LD}, \eqref{L5a} holds for $j=k-1$. This completes the proof.

\subsection{Proof of Lemma \ref{lemma:ij}} \label{proof:ij}
We introduce an auxiliary equation
\begin{align}
	L_{\alpha}(i,j)\le L_{\alpha}(i,j+1)<D_{\alpha}(i),\label{L7c}
\end{align}
and prove \eqref{L7a} and \eqref{L7c} for $j < j_{\alpha}$ by induction.

(a) Firstly, we prove \eqref{L7a} and \eqref{L7c} hold for $j=j_{\alpha}-1$.

(a.1) According to the definition of $j_{\alpha}$, we have $\mu_{\alpha}(i,j_{\alpha}-1)=0$ for $i\ge i_{0}$. Therefore, $A_{\alpha}^{0-2}(i,j_{\alpha}-1) \le A_{\alpha}^{0-2}(i_{0},j_{\alpha}-1)<0$ for $i<i_{0}$ as $A_{\alpha}^{0-2}(i,j)$ is non-decreasing in $i$ according to \eqref{eq:A02}. We have $\mu_{\alpha}(i,j_{\alpha}-1)\in \{0,1\}, i< i_{0}$.
According to Lemma \ref{lemma:jalpha}, \eqref{L5a} holds for $j = j_{\alpha}$. Then, we have $A_{\alpha}^{1-2}(i,j_{\alpha}-1)=A_{\alpha}^{1-2}(i,j_{\alpha})$ according to  \eqref{eq:A12l}. For $i\ge i_{\alpha}$, as \eqref{L4a} holds, we have $A_{\alpha}^{1-2}(i,j_{\alpha}-1)=A_{\alpha}^{1-2}(i,j_{\alpha})>0$, which means $\mu_{\alpha}(i,j_{\alpha}-1)\neq 1$. Therefore, we have $\mu_{\alpha}(i,j_{\alpha}-1)=0$ for any $i_{\alpha}\le i<i_{0}$.

(a.2) As $L_{\alpha}(i,j_{\alpha})<D_{\alpha}(i)$, similar to subsection (d) in the proof of Lemma \ref{lemma:jalpha}, we have $\mu_{\alpha}(i-1,j_{\alpha}-1)\neq 0$ if $\mu_{\alpha}(i,j_{\alpha}-1)=1$. As $\mu_{\alpha}(i,j_{\alpha}-1)\in \{0,1\}$, we have $\mu_{\alpha}(i-1,j_{\alpha}-1)=1$ if $\mu_{\alpha}(i,j_{\alpha}-1)=1$. Hence, there exists $I_{\alpha}(j_{\alpha}-1) \le i_{\alpha}$ so that \eqref{L7a} holds for $j=j_{\alpha}-1$.

(a.3) As \eqref{L7a} holds for $j=j_{\alpha}-1$, if $\mu_{\alpha}(i,j_{\alpha}-1) = 0$, we have $\mu_{\alpha}(i+1,j_{\alpha}-1) = 0$. Hence,
\begin{align*}
L_{\alpha}(i,j_{\alpha}-1) = \alpha L_{\alpha}(i+1,j_{\alpha}) < \alpha p L_{\alpha}(i+1,j_{\alpha})+\alpha \left( 1-p\right)D_{\alpha}(i+1)
\end{align*}
as $L_{\alpha}(i,j_{\alpha}) < D_{\alpha}(i)$. If $\mu_{\alpha}(i,j_{\alpha}-1) = 1$, as $J_{\alpha}(i,j)\le Q_{\alpha}(i,j,1)$, we have
\begin{align*}
L_{\alpha}(i,j_{\alpha}\!-\!1)\le Q_{\alpha}(i+1,j_{\alpha}\!-\!1,1)-Q_{\alpha}(i,j_{\alpha}\!-\!1,1)=\alpha p L_{\alpha}(i+1,j_{\alpha})+\alpha \left(1-p\right)D_{\alpha}(i+1).
\end{align*}
In both cases, we have
\begin{align*}
	L_{\alpha}(i,j_{\alpha}-1)&\le\alpha p L_{\alpha}(i+1,j_{\alpha})+\alpha\left(1-p\right)D_{\alpha}(i+1) \nonumber\\
	&=\alpha p L_{\alpha}(i+1,j_{\alpha}+1)+\alpha\left(1-p\right)D_{\alpha}(i+1)=L_{\alpha}(i,j_{\alpha})
\end{align*}
for $i<i_{\alpha}-1$ as $L_{\alpha}(i,j_{\alpha})=L_{\alpha}(i,j_{\alpha}+1)$ and $\mu_{\alpha}(i+1,j_{\alpha}) = \mu_{\alpha}(i,j_{\alpha}) = 1$ based on Lemma \ref{lemma:jalpha}.

For $i\ge i_{\alpha}$, $L_{\alpha}(i,j_{\alpha}-1)=\alpha\left( J_{\alpha}(i+2,j_{\alpha})-J_{\alpha}(i+1,j_{\alpha})\right)=0$ as $\mu_{\alpha}(i,j_{\alpha}-1)=0$ and $\mu_{\alpha}(i,j_{\alpha})=2$. Similarly,  $L_{\alpha}(i_{\alpha}-1,j_{\alpha}-1)=0$ if $\mu_{\alpha}(i_{\alpha}-1,j_{\alpha}-1)=0$. If $\mu_{\alpha}(i_{\alpha}-1,j_{\alpha}-1)=1$,
\begin{align*}
J_{\alpha}(i_{\alpha}\!-\!1,j_{\alpha})\!-\!J_{\alpha}(i_{\alpha}\!-\!1,j_{\alpha}\!-\!1) &= Q_{\alpha}(i_{\alpha}-1,j_{\alpha},1)-Q_{\alpha}(i_{\alpha}-1,j_{\alpha}-1,1) \nonumber\\
&= 1 + \alpha p (J_{\alpha}(i_{\alpha},j_{\alpha}+1)-J_{\alpha}(i_{\alpha},j_{\alpha})) \nonumber\\
&= 1 + \alpha p (J_{\alpha}(1,j_{\alpha}+1)-J_{\alpha}(1,j_{\alpha})) \nonumber\\
&= Q_{\alpha}(i_{\alpha},j_{\alpha},2)\!-\!Q_{\alpha}(i_{\alpha},j_{\alpha}\!-\!1,2) \nonumber\\
&\le J_{\alpha}(i_{\alpha},j_{\alpha})\!-\!J_{\alpha}(i_{\alpha},j_{\alpha}\!-\!1)
\end{align*}
as \eqref{L4c} holds for $j\ge j_{\alpha}$. Therefore, $L_{\alpha}(i_{\alpha}-1,j_{\alpha}-1)\le L_{\alpha}(i_{\alpha}-1,j_{\alpha})$. Hence, we have $L_{\alpha}(i,j_{\alpha}-1)\le L_{\alpha}(i,j_{\alpha})$ for all $i$. According to Lemma \ref{lemma:jalpha}, we have $L_{\alpha}(i,j_{\alpha}) < D_{\alpha}(i)$. Thus, \eqref{L7c} holds for $j = j_{\alpha}-1$.

(b) Suppose \eqref{L7a} and \eqref{L7c} hold for $j \ge k$, where $1 < k < j_{\alpha}$ and $I_{\alpha}(k)\le I_{\alpha}(k+1)\le \cdots \le i_{\alpha}$, consider $j=k-1$.

(b.1) As $\mu_{\alpha}(i,k)=0$ for $i\ge I_{\alpha}(k)$, as $A_{\alpha}^{0-1}(i,j)=\alpha(1-p)(J_{\alpha}(i+1,j+1)-J_{\alpha}(i+1,i+1))-\omega E_{\mathrm t}$ increases with $j$,we have $A_{\alpha}^{0-1}(i,k-1)<A_{\alpha}^{0-1}(i,k)<0$. Thus,  $\mu_{\alpha}(i,k-1)\in \{0,2\}$ for $i\ge I_{\alpha}(k)$. As \eqref{L7c} holds for $j=k$, we have
\begin{align*}
	A_{\alpha}^{0\!-\!2}(i,k\!-\!1)&=\alpha p\sum_{m=1}^{i}L_{\alpha}(m,k)\!+\!\alpha\left(1\!-\!p\right)\left(J_{\alpha}(i\!+\!1,k)\!-\!J_{\alpha}(1,1)\right)\!-\!\omega E_{\mathrm{t}}\!-\!\omega E_{\mathrm{s}} \nonumber\\
	&\le \alpha p\sum_{m=1}^{i}L_{\alpha}(m,k\!+\!1)\!+\!\alpha\left(1\!-\!p\right)\left(J_{\alpha}(i\!+\!1,k\!+\!1)\!-\!J_{\alpha}(1,1)\right)\!-\!\omega E_{\mathrm{t}}\!-\!\omega E_{\mathrm{s}} =A_{\alpha}^{0\!-\!2}(i,k),
\end{align*}
As \eqref{L7a} holds for $j=k$,  we have $A_{\alpha}^{0-2}(i,k-1) \le A_{\alpha}^{0-2}(i,k)<0$. Hence, we have $\mu_{\alpha}(i,k-1)=0$ for $i\ge I_{\alpha}(k)$. In addition, as $A_{\alpha}^{0-2}(i,j)$ in non-decreasing in $i$, we have $A_{\alpha}^{0-2}(i,k-1) \le A_{\alpha}^{0-2}(I_{\alpha}(k),k-1)<0$ for $i<I_{\alpha}(k)$. Thus, $\mu_{\alpha}(i,k-1)\in \{0,1\}$ for $i< I_{\alpha}(k)$.

(b.2) Similar to (a.2), we have $\mu_{\alpha}(i-1,k-1)\neq 0$ if $\mu_{\alpha}(i,k-1)=1$ as $L_{\alpha}(i,k)<D_{\alpha}(i)$. Thus, according to (b.1), there exists $I_{\alpha}(k-1)\le I_{\alpha}(k)$ so that \eqref{L7a} holds.

(b.3) We now prove \eqref{L7c} for $j=k-1$. By assumption, we have $L_{\alpha}(i,k) < D_{\alpha}(i)$. Then, we prove $L_{\alpha}(i,k-1) \le L_{\alpha}(i,k)$ for all three cases.
\begin{enumerate}[I.]
	\item If $\mu_{\alpha}(i,k)=\mu_{\alpha}(i+1,k)=0$, since $I_{\alpha}(i,k-1)\le I_{\alpha}(i,k)$, we have $\mu_{\alpha}(i,k-1)=\mu_{\alpha}(i+1,k-1)=0$. As \eqref{L7c} holds for $j=k$, we have
	\begin{align*}
		L_{\alpha}(i,k-1)=\alpha L_{\alpha}(i+1,k)\le \alpha L_{\alpha}(i+1,k+1)=L_{\alpha}(i,k).
	\end{align*}
	\item If $\mu_{\alpha}(i,k)=\mu_{\alpha}(i+1,k)=1$, there are two possible cases.
	\begin{enumerate}[i.]
		\item $\mu_{\alpha}(i,k-1)=\mu_{\alpha}(i+1,k-1)=0$, we have 		
		\begin{align*}
			L_{\alpha}(i,k-1)=\alpha L_{\alpha}(i+1,k)<\alpha p L_{\alpha}(i+1,k)+\alpha \left(1-p\right)D_{\alpha}(i+1).
		\end{align*}
		\item $\mu_{\alpha}(i,k-1)=1,\mu_{\alpha}(i+1,k-1)\in \{0,1\}$, as $J_{\alpha}(i,j)\le Q_{\alpha}(i,j,1)$, we have
		\begin{align*}
			L_{\alpha}(i,k\!-\!1) \le Q_{\alpha}(i\!+\!1,k\!-\!1,1)\!-\!Q_{\alpha}(i,k\!-\!1,1) =\alpha p L_{\alpha}(i\!+\!1,k)\!+\!\alpha \left(1\!-\!p\right)D_{\alpha}(i\!+\!1).
		\end{align*}
	\end{enumerate}
	In both cases, since $L_{\alpha}(i,k)\le L_{\alpha}(i,k+1)$, we have
	\begin{align*}
		L_{\alpha}(i,k-1)\le \alpha p L_{\alpha}(i+1,k+1)+\alpha\left(1-p\right)D_{\alpha}(i+1)=L_{\alpha}(i,k).
	\end{align*}
	\item $\mu_{\alpha}(i,k)=1,\mu_{\alpha}(i+1,k)=0$. We have $\mu_{\alpha}(i+1,k-1)=0$  as $I_{\alpha}(k-1)\le I_{\alpha}(k)$. Thus,
	$\mu_{\alpha}(i,k-1)$ has two possible values.
	\begin{enumerate}[i.]
		\item $\mu_{\alpha}(i,k-1)=0$, due to the optimality of $J_{\alpha}(i,j)$, we have
		\begin{align*}
			L_{\alpha}(i,k)&\ge Q_{\alpha}(i+1,k,0)-Q_{\alpha}(i,k,0)=\alpha L_{\alpha}(i+1,k+1) \nonumber\\
			&\ge \alpha L_{\alpha}(i+1,k)=Q_{\alpha}(i+1,k-1,0)-Q_{\alpha}(i,k-1,0)\nonumber\\
            &=L_{\alpha}(i,k-1).
		\end{align*}
		\item $\mu_{\alpha}(i,k-1)=1$, we have
		\begin{align*}
			&{~}L_{\alpha}(i,k)-L_{\alpha}(i,k-1) \nonumber\\
			=&{~}Q_{\alpha}(i+1,k,0)-Q_{\alpha}(i,k,1)-\left(Q_{\alpha}(i+1,k-1,0)-Q_{\alpha}(i,k-1,1)\right) \nonumber\\
			=&{~}\alpha\left(J_{\alpha}(i+2,k+1)-J_{\alpha}(i+2,k)\right)-\alpha p \left(J_{\alpha}(i+1,k+1)-J_{\alpha}(i+1,k)\right) \nonumber\\
			>&{~}\alpha \left(J_{\alpha}(i+2,k+1)-J_{\alpha}(i+2,k)-J_{\alpha}(i+1,k+1)+J_{\alpha}(i+1,k)\right) \nonumber\\
			=&{~}\alpha \left(L_{\alpha}(i+1,k+1)-L_{\alpha}(i+1,k)\right) \ge 0.
		\end{align*}
	\end{enumerate}
\end{enumerate}
For all possible cases, we proved that \eqref{L7c} hold for $j=k-1$. This completes the proof.

\subsection{Proof of Proposition \ref{thm:average}} \label{proof:average}
Firstly, according to Lemma \ref{lemma:jalpha} and Lemma \ref{lemma:ij}, the optimal policy for the problem \eqref{prob:discount} with the extended state set $\mathcal{A}_{\mathrm{ext}}$ is
\begin{align}
\mu_{\alpha}(i,j) = \left\{\begin{array}{ll}
0, &\textrm{if~} i \ge I_{\alpha}(j), j < j_{\alpha},\\
1, &\textrm{if~} i < I_{\alpha}(j), j < j_{\alpha} \textrm{~or~} i < i_{\alpha}, j \ge j_{\alpha},\\
2, &\textrm{if~} i \ge i_{\alpha}, j \ge j_{\alpha},
 \end{array} \right. \label{eq:mualpha}
\end{align}
where $I_{\alpha}(1)\le I_{\alpha}(2)\le \cdots\le I_{\alpha}(j_{\alpha}-1)\le i_{\alpha}$. Furthermore, as $A_{\alpha}^{0-1}(i,i) = -\omega E_{\mathrm t} < 0$, we have $\mu_{\alpha}(i,i) \in \{0, 2\}$ for all $i \ge 1$. Hence, we have $i_{\alpha} \le j_{\alpha}$. Therefore, the optimal policy for the problem \eqref{prob:discount} with state set $\mathcal{A}$ is the same as \eqref{eq:mualpha}.

Then, we consider the average cost problem \eqref{prob:avgmin}. Note that there must exist some state $(i', j')$ such that $\mu(i', j') = 2$ since otherwise, the average cost is infinity as the AoIR grows to infinity with probability 1, which is obviously not optimal. According to Proposition \ref{prop:converge}, there is a sequence of discount factors $\{\beta_n\}$ such that $\mu_{\beta_n}(i,j) \rightarrow \mu(i,j)$. Then, there is an integer $N'$ such that $\mu_{\beta_n}(i',j') = \mu(i',j') = 2$ for all $n \ge N'$. Since $\mu_{\beta_n}(i,j)$ is a threshold policy as in \eqref{eq:mualpha}, we have $\mu_{\beta_n}(i,j) = 2$ for all $i \ge i', j \ge j'$ and $n \ge N'$. Therefore, $\mu(i,j) = 2$ for all $i \ge i', j\ge j'$.

For each state $(i,j) \in \mathcal{A}' = \{(i,j) \in \mathcal{A}| j \le j'\}$, there exists an integer $N(i,j)$ such that $\mu_{\beta_n}(i,j) = \mu(i,j)$ for all $n \ge N(i,j)$. As $\mathcal{A}'$ is finite, by setting $N_{\max} = \max_{(i,j) \in \mathcal{A}'} N(i,j)$, we have for any $n \ge N_{\max}$, $\mu_{\beta_n}(i,j) = \mu(i,j)$ for all $(i,j) \in \mathcal{A}'$. As $\mu_{\beta_n}(i', j') = 2$, we have $j_{\beta_n} \le j'$. Therefore, there exist $\theta(1) \le \cdots \le \theta(\theta_{\mathrm r}) \le \theta_{\mathrm t} \le \theta_{\mathrm r} \le j'$ such that $\mu(i,j)$ is expressed as \eqref{eq:mu} for $(i,j) \in \mathcal{A}'$.

According to the above, for $n \ge N_{\max}$, we have $\mu_{\beta_n}(i,j')= \mu(i,j') = 1$ for $i < \theta_{\mathrm t}$ and $\mu_{\beta_n}(i,j') = \mu(i,j') = 2$ for $i \ge \theta_{\mathrm t}$. As $\mu_{\beta_n}(i,j)$ is a threshold policy, we have $\mu_{\beta_n}(i,j) = 1$ for all $i < \theta_{\mathrm t}, j \ge j'$ and $\mu_{\beta_n}(i,j) = 2$ for all $i \ge \theta_{\mathrm t}, j \ge j'$. Therefore, $\mu(i,j) = 1$ for all $i < \theta_{\mathrm t}, j \ge j'$ and $\mu(i,j) = 2$ for all $i \ge \theta_{\mathrm t}, j \ge j'$.

In summary, \eqref{eq:mu} holds for all $(i, j) \in \mathcal{A}$. This completes the proof.

\subsection{Proof of Theorem \ref{thm:closed}} \label{proof:closed}
As the second element of the state $(i,j) \in \tilde{\mathcal{A}}$ is unique, we can re-index the state $(i,j)$ by $j$ for convenience. The state transition probability from $j$ to $k$ can be expressed as
\begin{align*}
p_{jk} = \left\{ \begin{array}{ll} 1, & \textrm{if~} k = j+1 \textrm{~and~} 1 < j < \theta_{\mathrm r},\\
p, & \textrm{if~} k = j+1 \textrm{~and~} j \ge \theta_{\mathrm r},\\
1-p, & \textrm{if~} j = \theta_{\mathrm r} + m \theta_{\mathrm t} + k-1 \textrm{~and~} m \ge 0, 1 \le k \le \theta_{\mathrm t},\\
0, & \textrm{else}.
\end{array}\right.
\end{align*}

By solving $\pi = \pi \mathbf{P}$, where $\pi = (\pi_1, \pi_2, \cdots)$ with $\pi_j$ denoting the stationary probability of the $j$-th state, and $\mathbf{P}$ is the state transition matrix with elements $p_{jk}, \forall j \ge 1, k \ge 1$, we have
\begin{align}
\pi_j = \left\{ \begin{array}{ll} \frac{1-p^j}{1-p^{\theta_{\mathrm t}}}\pi_0, & 1 \le j \le \theta_{\mathrm t}, \\
\pi_{0}, & \theta_{\mathrm t} < j \le \theta_{\mathrm r}, \\
p^{j-\theta_{\mathrm r}}\pi_{0}, & j > \theta_{\mathrm r}, \end{array}\right. \label{eq:pi_j}
\end{align}
where $\pi_0 = \dfrac{1-p^{\theta_{\mathrm t}}}{\theta_{\mathrm r}(1-p^{\theta_{\mathrm t}})+\theta_{\mathrm t}p^{\theta_{\mathrm t}}}$. Then, the average AoIR and the average energy consumption can be calculated as
\begin{align}
\bar{\Delta} &= \sum_{j=1}^{\infty} j\pi_j + \frac{1}{2}, \label{eq:delta}\\
\bar{E} &=  \sum_{m = 0}^{\infty}(E_{\mathrm t} + E_{\mathrm s})\pi_{\theta_{\mathrm r} + m\theta_{\mathrm t}} + \sum_{m = 0}^{\infty}\sum_{n=1}^{\theta_{\mathrm t}-1}E_{\mathrm t}\pi_{\theta_{\mathrm r} + m\theta_{\mathrm t}+n}. \label{eq:energy}
\end{align}
Substituting \eqref{eq:pi_j} into \eqref{eq:delta} and \eqref{eq:energy}, we can obtain \eqref{eq:aoiclosed} and \eqref{eq:eclosed}. This completes the proof.

\bibliographystyle{IEEEtran}
\bibliography{ref}

\end{document}